\documentclass[11pt, a4paper]{article}
\usepackage[top=2.4cm, bottom=2.4cm, left=2.4cm, right=2.4cm, footnotesep=1cm]{geometry}
\usepackage{amsmath,amsfonts}
\usepackage{authblk}
\usepackage[colorlinks,linkcolor=black,anchorcolor=black,citecolor=black,urlcolor=black]{hyperref}
\usepackage{array}
\usepackage{textcomp}
\usepackage{amsthm}
\usepackage{subcaption}
\usepackage[shortlabels]{enumitem}
\usepackage{stfloats}
\usepackage{url}
\usepackage{verbatim}
\usepackage{graphicx}
\usepackage{palatino}
\usepackage{newpxtext}
\usepackage{newpxmath}
\usepackage[ruled,vlined,linesnumbered]{algorithm2e}
\usepackage{setspace}
\usepackage[titletoc,title]{appendix}

\usepackage{natbib}
 \bibpunct[, ]{(}{)}{,}{a}{}{,}%

\newcommand{\mca}[0]{\mathcal{A}}
\newcommand{\mcs}[0]{\mathcal{S}}
\newcommand{\mcp}[0]{\mathcal{P}}

\newcommand{\mcd}[0]{\mathcal{D}}

\let\Contentsline\contentsline
\renewcommand

\newtheorem{definition}{Definition}
\newtheorem{lemma}{Lemma}
\newtheorem{proposition}{Proposition}
\newtheorem{assumption}{Assumption}
\newtheorem{remark}{Remark}
\newtheorem{theorem}{Theorem}

\title{\begin{spacing}{1.2}From Observation to Intervention: Mastering Bus Holding Control from Historical Uncontrolled Trajectories with Offline Reinforcement Learning\end{spacing}}

\author[1]{Minyu Shen}
\author[2]{Xiaowen Bi}
\author[3]{Weihua Gu}
\author[4]{Zhandong Xu}
\author[5]{Yuezhong Wu}
\author[6]{Feng Xiao}
\author[6]{Lu Li}
\affil[1]{School of Management Science and Engineering, Southwestern University of Finance and Economics, Chengdu, China}
\affil[2]{Department of Statistics and Data Science, Beijing Normal University-Hong Kong Baptist University United International College, Zhuhai, China}
\affil[3]{Department of Electrical and Electronic Engineering, The Hong Kong Polytechnic University, Hong Kong, China}
\affil[4]{School of Transportation and Logistics, Southwest Jiaotong University, Chengdu, China}
\affil[5]{College of Computer and Data Science, Fuzhou University, Fuzhou, China}
\affil[6]{Business School, Sichuan University, Chengdu, China}
\date{}

\hypersetup{
  pdftitle={From Observation to Intervention: Mastering Bus Holding Control from Historical Uncontrolled Trajectories with Offline Reinforcement Learning},
  pdfauthor={Minyu Shen, Xiaowen Bi, Weihua Gu, Zhandong Xu, Yuezhong Wu, Feng Xiao, Lu Li},
  pdfsubject={Offline reinforcement learning for bus holding control},
  pdfkeywords={Bus operations, Urban computing, GPS trajectories, Offline reinforcement learning, Multi-agent reinforcement learning, Graph neural network}
}

\begin{document}
\maketitle
\onehalfspacing

\noindent \textbf{Abstract:} Bus operations are inherently unstable---buses initially dispatched at equal headways will inevitably end up bunching up. This issue exacerbates passenger waiting times, creates load imbalances between buses, and undermines service reliability. Reinforcement learning (RL) has proven highly effective for mitigating bus bunching and outperforms traditional model-based approaches. However, existing RL approaches require agents to continuously interact with the environment, which potentially disrupts the bus system before effective policies are learned. This paper proposes an Offline Graph-based Multi-agent Actor-Critic (Offline-GMAC) method that trains bus holding policies using only readily available logged GPS data collected during normal uncontrolled operations, thereby eliminating the need for disruptive online interaction or exploration. The method employs general graph neural networks for multi-agent credit assignment. A key innovation of our method is a data augmentation scheme that synthesizes unperformed holding actions from historical uncontrolled bus trajectories. The scheme exploits the equivalence between stochastic link travel time delays and deliberate holding interventions. Under mild assumptions, we prove that the 1-Wasserstein distance between synthesized and target state transition distributions is bounded above by a term that scales linearly with the standard deviation of link travel times plus a tunable synthesis-dependent term. Evaluations using real-world bus data demonstrate that the proposed offline method achieves performance comparable to, and in some cases, even better than its online counterpart.
\\
\textbf{Keywords:} Bus operations, Urban computing, GPS trajectories, Offline reinforcement learning, Multi-agent reinforcement learning, Graph neural network
\section{Introduction}

\subsection{Background}
Public transportation is a necessary component of the urban transportation system for promoting sustainable and equitable transport. Among the various modes of public transportation, buses are one of the most widely used due to their extensive reach and flexibility in connecting diverse and vast urban regions \citep{APTA2024}. In 2022, approximately 5,835 billion passenger kilometers were traveled by bus in urban areas worldwide, with this number expected to reach 10,993 billion passenger kilometers by 2050 \citep{statista_bus_demand_2023}. As urban populations continue to grow, high-quality bus services are becoming increasingly critical to improving the quality of life for city dwellers.\par

Despite their importance, bus services often face criticism for poor reliability, with bus bunching being a persistent global issue. In an ideal scenario, buses would arrive at stops according to planned regular headways, thereby efficiently serving passengers. However, in reality, various factors disrupt this plan. Random traffic events (e.g., road construction, traffic accidents, or sudden passenger influxes at certain stops) may delay buses. These delayed buses encounter more passengers at subsequent stops, leading to longer service times and further delays. Conversely, buses that arrive early at a stop meet fewer passengers, allowing them to travel faster and eventually catch up with delayed ones \citep{newell1964maintaining,daganzo2009headway}. This self-reinforcing cycle occurs pervasively in real bus systems, as shown by the time-space bus trajectories on a typical route in Chengdu, China in Fig.~\ref{fig:vis_example}. As illustrated in this figure, many buses that were initially dispatched at nearly even headways eventually bunch at downstream stops.

\begin{figure}[!ht]
    \centering
    \includegraphics[scale=0.32]{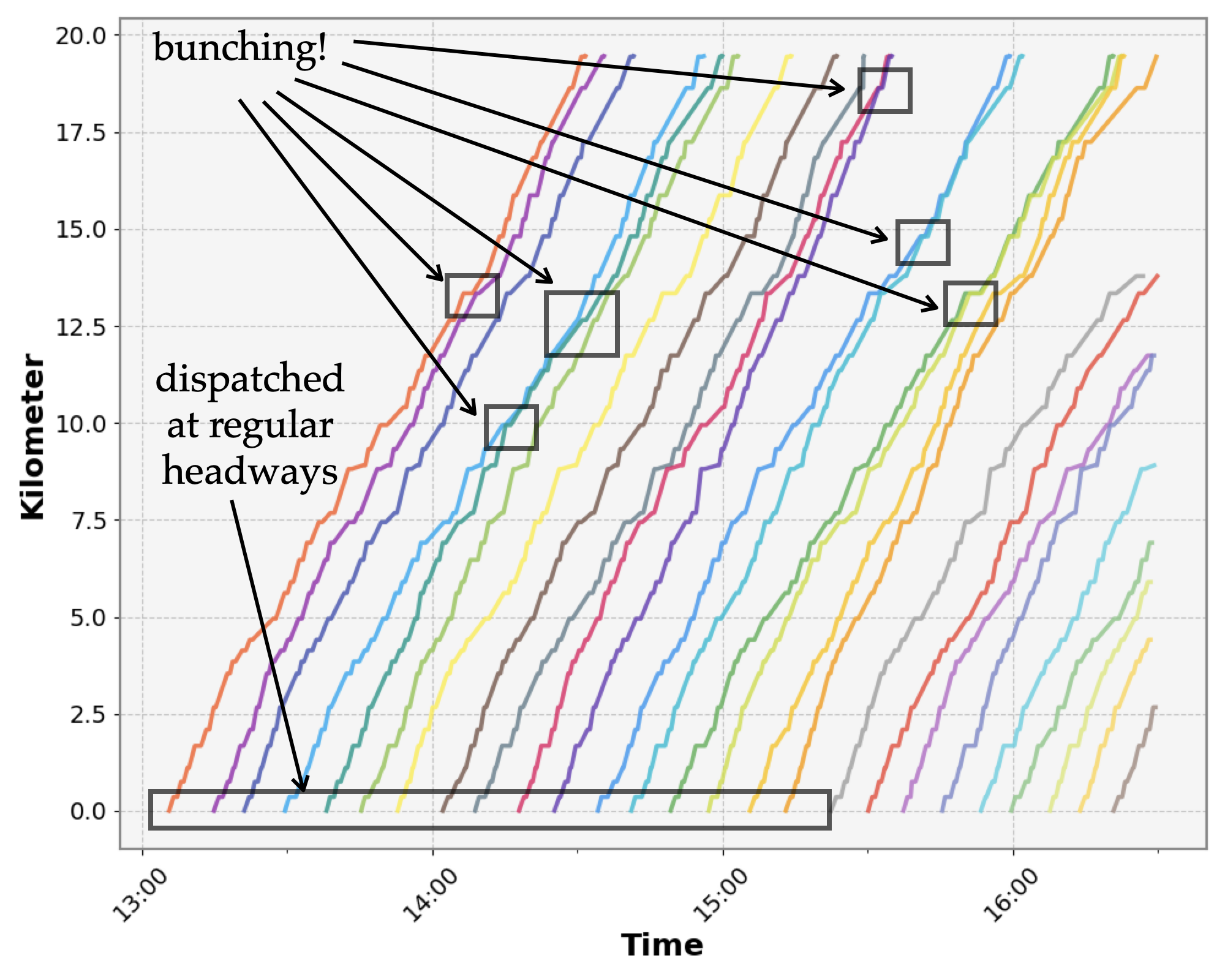}
    \caption{Time-Space diagram of bus trajectories of route 3 in Chengdu, China on March 8, 2021.}\label{fig:vis_example}
\end{figure}

Bus bunching leads to several critical issues: (i) \textit{Increased passenger waiting times}: Irregular headways between bus arrivals result in extended waits for passengers arriving during longer intervals, significantly increasing average waiting time \citep{osuna1972control}; (ii) \textit{Overcrowding}: Buses arriving after long headways become overcrowded, causing passenger discomfort and potentially preventing boarding on near-capacity buses; and (iii) \textit{Operational inefficiencies}: Bus bunching increases both dwell times and running times, reducing commercial speed and requiring larger terminal and stop areas for bus queuing \citep{verbich2016have}. These issues degrade service reliability and drive passengers to alternative transport modes.

Bus holding is the most widely adopted control strategy to combat bus bunching because of its operational simplicity \citep{newell1964maintaining}. The basic idea is that when the gap between two consecutive buses narrows, the following bus is intentionally held at a designated stop for a short period beyond its service time to restore the regular spacings. Although the holding operation itself is straightforward, implementing effective holding control in real-world scenarios is difficult due to: (i) complex, unpredictable environments with variable traffic conditions, passenger demand, and driver behaviors; (ii) interconnected nature of holding decisions where interventions ripple through the entire route; and (iii) need for rapid real-time decision-making in response to changing conditions.\par

\subsection{Motivations and challenges}
To address the above control difficulties, recent studies have focused on model-free reinforcement learning (RL) techniques to train bus holding control policies that outperform traditional model-based approaches in simulation environments \citep{alesiani2018reinforcement,wang2020dynamic,wang2021reducing}. While model-based methods struggle to accurately capture the complex dynamics of bus operations due to simplifying assumptions and difficulties in estimating time-varying parameters, RL methods can learn directly from environmental interactions without relying on predefined models. Current RL methods follow an online learning paradigm where agents iteratively experiment with various actions using the latest learned policy, then leverage that experience to refine the policy \citep{levine2020offline}. This direct interaction enables RL to adapt to changing conditions and learn complex relationships between states, actions, and outcomes that may be difficult to model explicitly.\par

However, despite the state-of-the-art performance achieved by online RL, it relies on online exploration, which is often impractical in real-world settings. Such exploration can lead to increased passenger complaints and operating costs before agents develop effective policies. To address this practical limitation, we propose leveraging historical bus GPS trajectory data routinely collected by transit agencies for RL training. These GPS logs typically provide sufficient granularity to extract bus locations at frequent intervals, as well as the arrival and departure times at each stop for most trips. These data repositories contain valuable information about bus operations under various conditions.\par

Traditionally, such trajectory data have been utilized primarily for performance measurement and travel time prediction \citep{li2023sequence}. However, their application for RL training remains largely unexplored and faces three significant challenges: (i) \textit{Missing action labels}: Historical GPS trajectories lack information on holding actions, as they are collected without implemented controls. Inferring the effects of different actions without environment interaction is crucial for obtaining the state-action-reward pairs necessary for RL training; (ii) \textit{Asynchronous multi-agent interactions}: Bus holding actions occur asynchronously at different stops but collectively influence the state transitions of all buses in the system. Effectively capturing these complex interactions from historical data is challenging; and (iii) \textit{Offline hyperparameter selection}: Historical data does not provide an execution environment, which makes it difficult to evaluate policies under different hyperparameters. Identifying optimal hyperparameter configurations without real environment interaction is therefore also challenging.\par

We propose to address these challenges by transforming historical uncontrolled trajectories into optimal control policies through the following pipeline (Fig.~\ref{fig:overall_pipeline}). First, to overcome the absence of action labels, we exploit a fundamental equivalence between stochastic travel time deviations and deliberate holding interventions. In uncontrolled operations, certain spontaneous events (e.g., extra dwell times due to passenger crowding or experienced drivers’ speed adjustments) can incidentally improve headway regularity, mimicking the effects of intentional holds. By treating these stochastic delays as implicit holding actions, we synthesize unperformed holding actions from historical uncontrolled GPS data. This process yields the essential state-action-reward-next-state transition tuples for training. Second, to capture asynchronous multi-agent interactions, we construct upstream and downstream event graphs that encode spatiotemporal dependencies among buses, enabling a global critic to estimate the system-wide impact of local decisions via message passing. Third, leveraging these synthesized tuples and event graphs, we develop an offline RL training scheme to learn optimal holding strategies. This pipeline bridges the gap between uncontrolled historical observations and decision-oriented control policies.

\subsection{Contributions of this paper}
In this paper, we propose a data-driven, offline multi-agent RL approach for learning optimal bus holding policies directly from uncontrolled bus trajectory logs. (Code and data are available at: \url{https://github.com/Minyu-Shen/OGMAC}.) Our major contributions are:
\begin{enumerate}
	\item We propose an Offline Graph-based Multi-agent Actor-Critic (Offline-GMAC) method, marking the first work to train a bus holding policy using only logged GPS data, without directly interacting with the environment. Compared to the existing online RL methods, the key advantage of our offline approach is that it avoids potentially harmful real-world trials for training purposes, while still leveraging the strengths of model-free RL.
	\item We propose a data augmentation scheme based on the insight that stochastic link travel time delays are functionally equivalent to deliberate holds. Exploiting this equivalence, the scheme ``synthesizes'' \textit{unperformed} holding actions from historical GPS trajectories of \textit{uncontrolled} bus motions. To validate this scheme, we prove that under mild assumptions, the 1-Wasserstein distance between the synthesized state transition distribution (from uncontrolled trajectories) and the desired distribution (from actual holding control) is bounded above by a term that scales with the standard deviation of link travel times plus a tunable synthesis-dependent term. The scheme eliminates the need for real-world training data under various control actions, which are very costly to gather, and removes the necessity of building a simulator as a surrogate model for collecting training data.
	\item We propose an offline hyperparameter selection (OHS) approach that utilizes the critic value over offline training iterations as the scalar statistic for selecting optimal hyperparameter configurations. This method can identify high-performing policies from a set of policy candidates trained using various hyperparameters, requiring only historical uncontrolled data as input.
	\item Experiments using real data from a bus route in Chengdu, China, show that Offline-GMAC can learn a policy that performs on par with---and in some cases surpasses---a finely-tuned online RL policy using under half a month's worth of historical, uncontrolled bus trajectories. This finding unlocks the vast potential of historical bus trajectory data, which has been largely unused for real-time control in current research and practices.

\end{enumerate}

Section \ref{sec:related_work} summarizes the related work regarding bus holding strategies. Section \ref{sec:problem_description} describes the problem setup and the asynchronous issue of the bus holding problem. Section \ref{sec:method} formulates the bus holding problem as a multi-agent RL problem and proposes the Offline-GMAC method together with the training data synthesis scheme. Section \ref{sec:experiments} examines the performance of the proposed method via numerical experiments based on real-world bus line operation data. Matters on the real-world deployment of the Offline-GMAC are discussed in Section \ref{sec:deployment}. Conclusions are summarized in Section \ref{sec:conclusion}.

\section{Literature review on bus holding strategies}\label{sec:related_work}
The literature on bus holding can be categorized into two main streams: model-based holding strategies and model-free, RL-based holding strategies. Model-based strategies are devised and analyzed using mathematical models, which are formulated to characterize the essential operating factors and dynamics of bus systems. RL-based methods, on the other hand, learn an optimal control policy through iterative trial-and-error interactions with the environment, without requiring an explicit model of the environment.

\subsection{Model-based holding strategies}
Earlier works on bus holding have focused on schedule-based holding strategies where fixed slacks were inserted into bus schedules and buses were required to depart on time at pre-determined control points \citep{osuna1972control,newell1974control,zhao2006optimal}.\par

Since the early 2000s, with the development of GPS tracking systems on public transit vehicles, it has become easier to efficiently collect real-time data on bus locations and speeds. This technological advance facilitates the implementation of real-time adaptive holding strategies in practice. \citet{daganzo2009headway} introduced a dynamic holding strategy applied at multiple control points where holding times were determined using the bus headways measured in real-time. This work was later extended to a variety of adaptive holding strategies \citep{daganzo2011reducing,xuan2011dynamic,bartholdi2012self,he2015anti,andres2017predictive,zhang2018two,dai2019predictive}. Works in this realm often model the system by difference equations incorporating random noise, and they conduct analytical assessments of the stability of the proposed holding strategy. To maintain analytical tractability, simplifying assumptions have been adopted, preventing these models from incorporating realistic elements like bus capacity constraints, random passenger arrivals, diverse passenger behaviors, and signalized intersections.\par

In addition to the aforementioned adaptive methods, a significant body of research has proposed rolling-horizon-based methods, also known as model-predictive control (MPC) methods. These methods involve repeatedly solving an optimization problem over a finite time horizon, aiming to optimize a performance measure while adhering to constraints on the system's state and control inputs. Common objectives include minimizing expected passenger waiting time or bus headway variations. For instance, \citet{hickman2001analytic} presents an analytical method based on a stochastic transit service model to determine the optimal bus holding time at a control stop. \citet{eberlein2001holding} formulate the control problem as a rolling-horizon quadratic programming problem and find the single-point holding time that minimizes passenger waiting time. \citet{zhao2003distributed} propose a cooperative holding control algorithm that assumes communication between vehicles and stops to reduce passenger waiting time. \citet{delgado2012much} consider more realistic operational details, such as heterogeneous bus fleets and limited bus capacity, and propose a control method that combines skip-stop and holding strategies. \citet{berrebi2015real} develop an analytical method to determine holding times that minimize headway variations based on the prediction of future bus arrival times. \citet{hernandez2015analysis} develop control schemes based on holding strategies for a corridor with multiple bus lines. They consider scenarios where a central operator aims to maximize the service level for all bus lines, as well as situations where each bus line operates independently. \citet{seman2019headway} develop a mathematical programming approach to find optimal holding times that minimize passengers' travel times for a corridor shared by multiple lines.

\subsection{Model-free RL-based holding strategies}
Model-based methods struggle to accurately capture real-world bus operations. Their simplifying assumptions and difficulties in estimating time-varying parameters like origin-destination (OD) demand, inter-stop travel times, and dwell times can compromise control effectiveness in actual deployment.\par

To address these issues, recent studies propose model-free RL methods to mitigate bus bunching. RL methods learn optimal policies through direct interaction with the bus operating environment, thereby eliminating the need for explicit system models. These approaches iteratively refine policies by collecting and learning from interactions with the environment \citep{sutton2018reinforcement}. For example, \citet{chen2016real} present a multi-agent Q-learning architecture to optimize the holding times for coordinated bus operations. \citet{alesiani2018reinforcement} propose an RL model for solving the multi-agent bus-holding problem where the agents share a common Q-function and a target headway is specified. \citet{wang2020dynamic} model the holding problem as a Markov game and design a joint action tracker to better quantify the interactions among multiple agents. \citet{wang2021reducing} recognize the asynchronous nature of bus holding decisions and extend the classical actor-critic architecture by introducing an inductive critic that evaluates the impacts of asynchronous actions from other agents. \citet{he2022dynamic} propose an adaptive holding method combining approximate dynamic programming (ADP) and a multistage look-ahead mechanism to determine the holding time. The incorporation of the latter facilitates a faster and smoother transition through the algorithm's initial unstable phase. \citet{shen2024extensible} developed an extensible open-source simulation platform to facilitate the systematic development and benchmarking of different bus holding strategies. \citet{xu2025reinforcement} conducted a comprehensive and systematic evaluation of various reinforcement learning configurations to examine their control performance for mitigating bus bunching. These studies generally report that online RL methods outperform the traditional ones. In a recent study, \citet{wang2023multi} develop an efficient multi-agent RL framework to learn bus holding control in a corridor shared by two bus lines.\par

However, the nature of RL also presents its own set of challenges. RL depends on a trial-and-error approach, where agents iteratively try various actions by interacting with the environment with the latest learned policy and then leverage that experience to improve the policy \citep{levine2020offline}. However, in real-world bus holding control, this method of online exploration is often impractical. It can lead to an increase in passenger complaints and operating costs before the agents develop effective policies. This limitation partially explains why RL methods find their greatest success in areas where high-fidelity simulators are available, such as Atari games \citep{mnih2013playing}, Go \citep{silver2017mastering}, and Dota 2 \citep{berner2019dota}. These environments allow for extensive exploration and learning without the real-world repercussions associated with online training. Unfortunately, such a high-fidelity simulator does not yet exist for complex bus systems.

\section{Problem description}\label{sec:problem_description}
Consider a one-way loop corridor that spans a distance of $X$ units. The corridor includes a terminal stop and is served by a bus fleet of size $M$, $B\equiv\{ b_1,b_2,\ldots,b_M \}$; see Fig.~\ref{fig:corridor_layout}. Buses are labelled sequentially with bus $b_{j\oplus 1}$ following $b_{j}$ and $b_{j\ominus 1}$ preceding $b_{j}$. Here, the ``o-plus'' and ``o-minus'' notations denote addition and subtraction modulo $M$ for the bus number index, respectively. Each bus $b_j\in B$ operates in a loop along the corridor, sequentially visiting a set of $N$ stops, $C\equiv\{c_1,c_2,\ldots,c_N\}$. The stops are numbered sequentially along the corridor, starting from the terminal stop. When a bus reaches the terminal stop, it either immediately begins a new loop or undergoes a layover time. An inter-stop link connecting two stops $c_k$ and $c_{k\oplus 1}$ is denoted as $l_k,l_k\in L\equiv \{l_1,l_2,\ldots,l_N\}$, where $L$ is the link set. In this context, the ``o-plus'' and ``o-minus'' notations perform addition and subtraction modulo $N$, respectively.
\begin{figure}[!ht]
    \centering
    \includegraphics[scale=0.28]{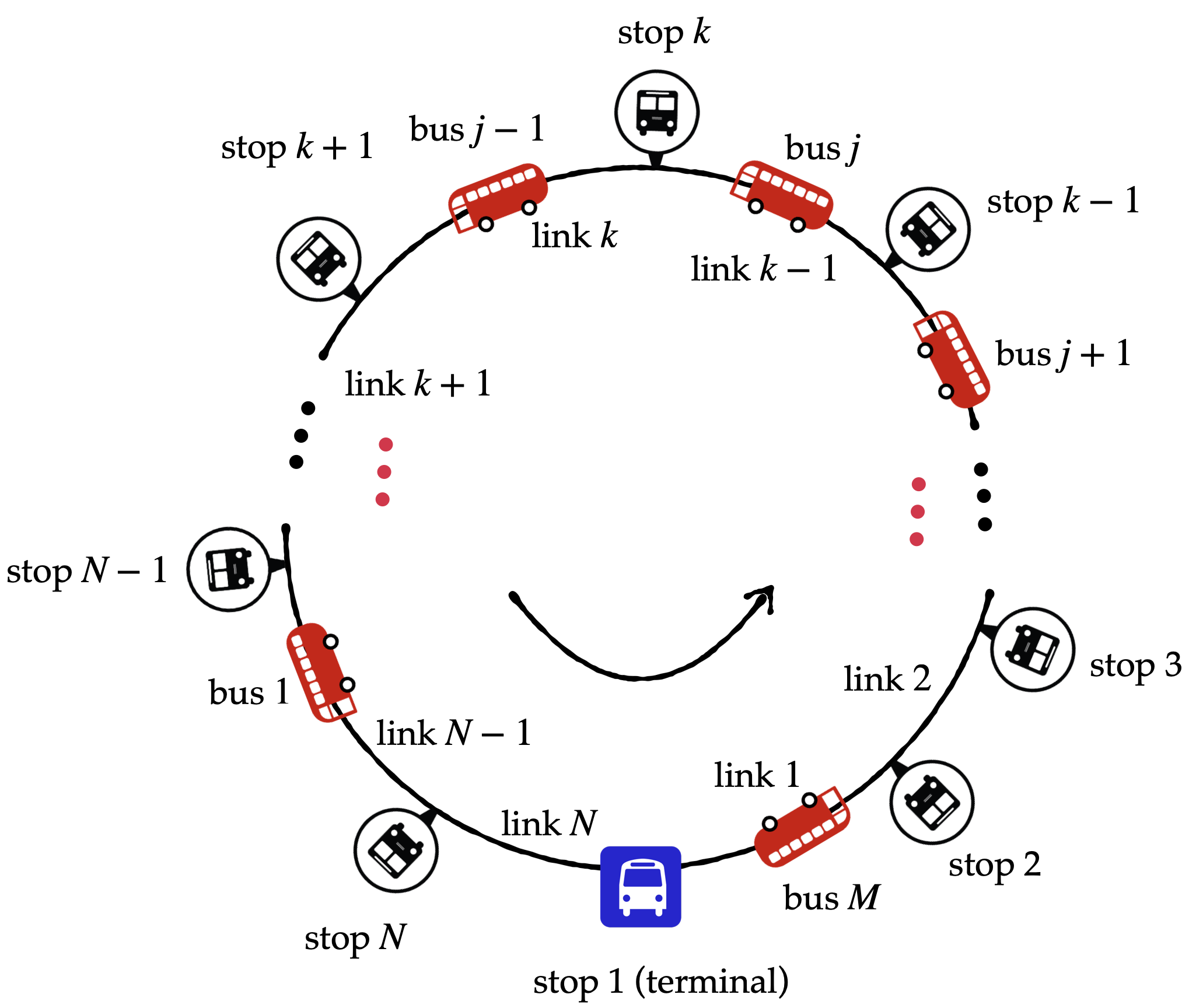}
    \caption{Layout of a one-way loop corridor.}\label{fig:corridor_layout}
\end{figure}\par

Bus motion tends to become unstable over time, leading to uneven bus spacing and bunching. Bus holding is an effective solution to this problem. In this paper, we propose a cooperative multi-agent RL framework to develop bus holding policies. Each agent represents an individual bus, and they cooperate to hold in order to maintain even spacings as much as possible. To leverage the most up-to-date information, we determine the holding control when buses depart from bus stops, in contrast to previous studies that determined control when buses arrive at stops (e.g., \citealp{daganzo2009headway} and \citealp{wang2020dynamic}). \citet{gu2025nonlinear} showed that determining the holding time when a bus is ready to depart consistently reduces headway variability for the same average holding delay, compared to strategies that make the decision upon the bus's arrival. Additionally, the holding control can be implemented by slowing down buses on links, catering to impatient passengers who may find it difficult to tolerate holding delays.\par

Holding control for a bus fleet is a typical application of asynchronous multi-agent control, because the bus agents rarely implement actions at the same time \citep{wang2021reducing}. To illustrate, Fig.~\ref{fig:ts_illustrate} depicts a time-space diagram of three consecutive bus trajectories under holding. In the figure, $t^*$ represents the decision time for bus $b_j$ to determine its holding time when departing from stop $c_k$. Notably, at $t^*$, neither the preceding bus $b_{j\ominus 1}$, dwelling at stop $c_{k\oplus 2}$, nor the succeeding bus $b_{j\oplus 1}$, cruising along link $l_{k\ominus 2}$, are at their respective decision times.
\begin{figure}[!ht]
    \centering
    \includegraphics[scale=0.28]{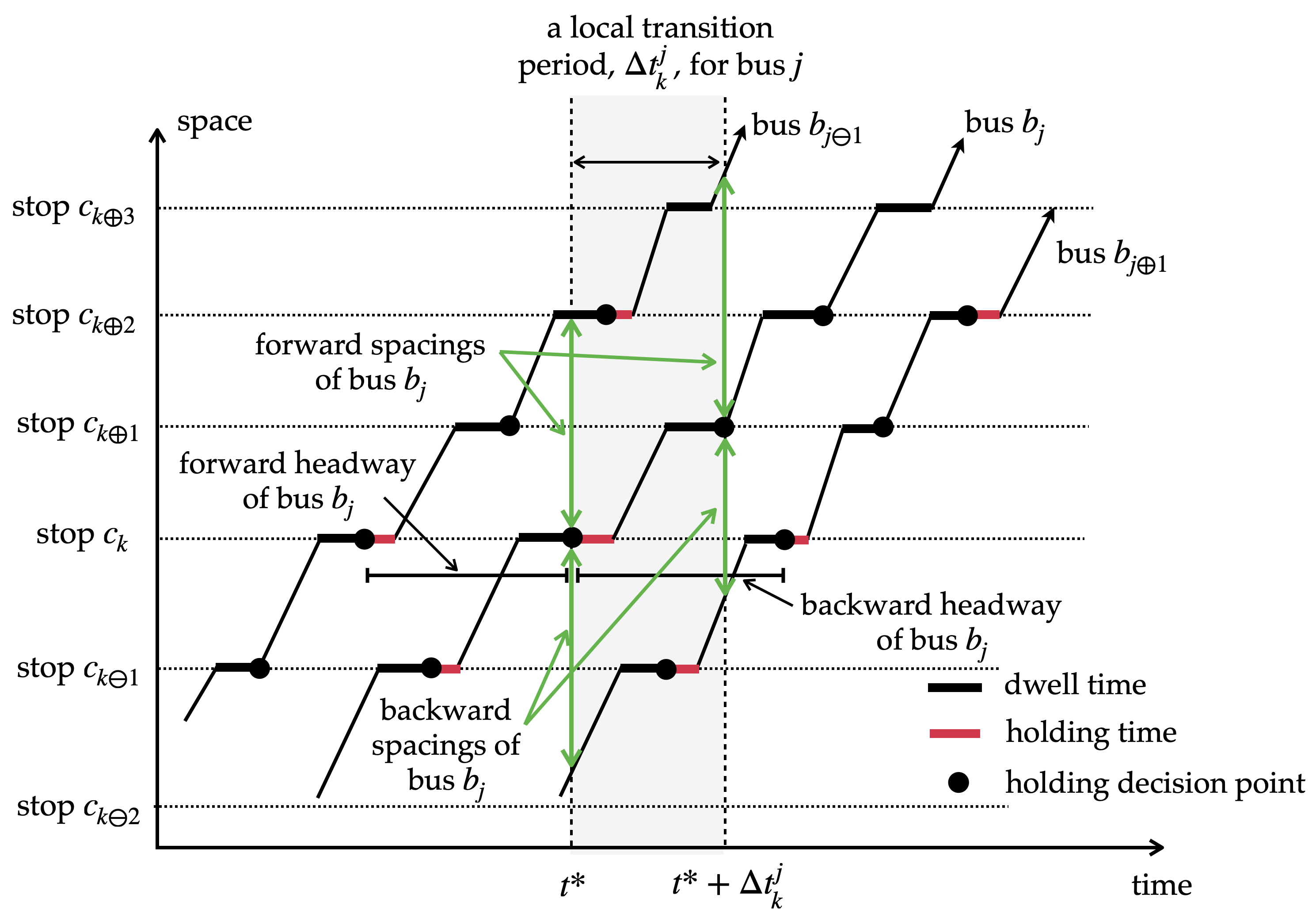}
    \caption{Time-space diagram illustration of bus holding.}\label{fig:ts_illustrate}
\end{figure}\par

While the actions of other agents may not occur simultaneously, they collectively influence both the state transition of the current agent and the overall performance of the system. For instance, in Fig.~\ref{fig:ts_illustrate}, the subsequent spacing between buses $b_{j\ominus 1}$ and $b_j$ (a state variable for bus agent $b_j$), measured at the next stop $c_{k\oplus 1}$, is affected by both the current action of bus $b_j$ at stop $c_k$ and the action taken by bus $b_{j\ominus 1}$ at stop $c_{k\oplus 2}$.\par

Inspired by the idea of graph attention neural network in \citet{wang2021reducing}, we develop a more general multi-agent RL approach capable of leveraging diverse graph neural network (GNN) architectures to effectively handle asynchrony and capture inter-agent interactions. More importantly, our proposed \textit{offline} multi-agent RL method can eliminate the need for real-time interactions with the environment during the training process, as it solely relies on historical uncontrolled bus data to learn an optimal holding policy.\par

\section{Offline Graph-based Multi-agent Actor-Critic Method}\label{sec:method}
In this section, we introduce an Offline Graph-based Multi-agent Actor-Critic (Offline-GMAC) method for learning an optimal holding policy from a historical bus trajectory dataset. Our method offers an end-to-end solution that eliminates the need for constructing a high-fidelity simulator for training purposes. The method consists of two key components: a data preparation process and an off-policy multi-agent training scheme. The data preparation process transforms historical uncontrolled trajectories into graph-based structures, resulting in usable transition tuples. These tuples are then utilized in the off-policy multi-agent training scheme to learn optimal policies. Fig. \ref{fig:overall_pipeline} illustrates the overall pipeline and data flow of our method, while Fig. \ref{fig:network_architecture} showcases the neural network architectures used for approximating value functions. The RL model involves an actor that makes holding decisions based on local states, a local critic that evaluates the actor's performance, and a global critic that utilizes a GNN to estimate the impact of holding events from all the other agents.\par

Section \ref{sec:rl_elements} describes the state, action, local transition dynamics, and reward designs. The policy and value function approximations using neural networks are presented in Section \ref{sec:value_functions}. Section \ref{sec:training_scheme} illustrates the training scheme for learning the optimal policy from historical data. Crucially, we demonstrate the process of synthesizing uncontrolled bus trajectories into controlled training data in Section \ref{sec:data_synthesis}. Theoretical justifications for this synthesis scheme are provided in Section \ref{sec:theoretical_justification_data_synthesis}.

\subsection{RL basic elements design}\label{sec:rl_elements}
We begin by defining the local state transition of each bus agent $b_j$, from the departure at stop $c_k$ to the subsequent departure at the next stop $c_{k\oplus 1}$: $\widetilde{\mcp}^j: \mcs^j \times \mca^j \mapsto \mcs^j$, where $\mcs^j$ denotes the local state space and $\mca^j$ represents the action space. An actor selects an action based on the local state, following a holding policy. Furthermore, a local critic assesses the quality of the action by considering the local state and the actor's chosen action. Note that the above local state transition and bus agent $b_j$'s action selection do not yet consider the impacts of other agents' actions. The consideration of multi-agent interactions will be facilitated by the global event critic, as discussed in Section \ref{sec:global_critic}.\par

\subsubsection{Local state}
With an onboard GPS device or other location-tracking technologies, bus $b_j$'s location at each time $t$, denoted by $x^{j}_{t}$, can be easily tracked. When bus $b_j$ departs from stop $c_k$ at time $t$, its local state $s^j_{k,t}$ consists of the forward spacing $f^{j}_{k,t}\equiv (x^{j\ominus 1}_t \ominus x^{j}_{t})/X$ (representing the spacing between buses $b_j$ and $b_{j\ominus 1}$, normalized by the corridor length $X$), the backward spacing $b^{j}_{k,t}\equiv (x^{j}_{t} \ominus x^{j\oplus 1}_t)/X$ (between buses $b_j$ and $b_{j\oplus 1}$), and the current departure stop ID, $k$. The observed spacings are depicted by the green vertical segment lines in Fig.~\ref{fig:ts_illustrate}. Note that when applying the ``o-minus'' notation on bus positions, it signifies the operation of subtraction modulo $X$. The last state variable is included to account for the heterogeneity across bus stops and links.\par

In our approach, we utilize spacings as state variables instead of the commonly-used headways in model-based studies (e.g., \citealp{daganzo2009headway}). This is because the forward headway of bus $b_j$ (as shown in Fig.~\ref{fig:ts_illustrate}) does not reflect the up-to-date position of bus $b_{j\ominus 1}$. Meanwhile, the backward headway (also depicted in Fig.~\ref{fig:ts_illustrate}) cannot be accurately measured at present. Conversely, the forward and backward spacings of bus $b_j$ reflect the accurate and up-to-date states about the three consecutive buses, $b_{j\ominus 1}$, $b_j$, and $b_{j\oplus 1}$, at a specific time $t$.\par

Our method aims to learn an optimal policy directly from historical bus GPS trajectories. In other words, the RL elements will be extracted from the GPS data. Hence, variables that are not included in the GPS data, such as the number of passengers on board \citep{wang2021reducing} and the number of passengers waiting at the stop \citep{wang2020dynamic}, are not included in the system state.
\begin{figure*}[!ht]
    \centering
    \includegraphics[width=\textwidth]{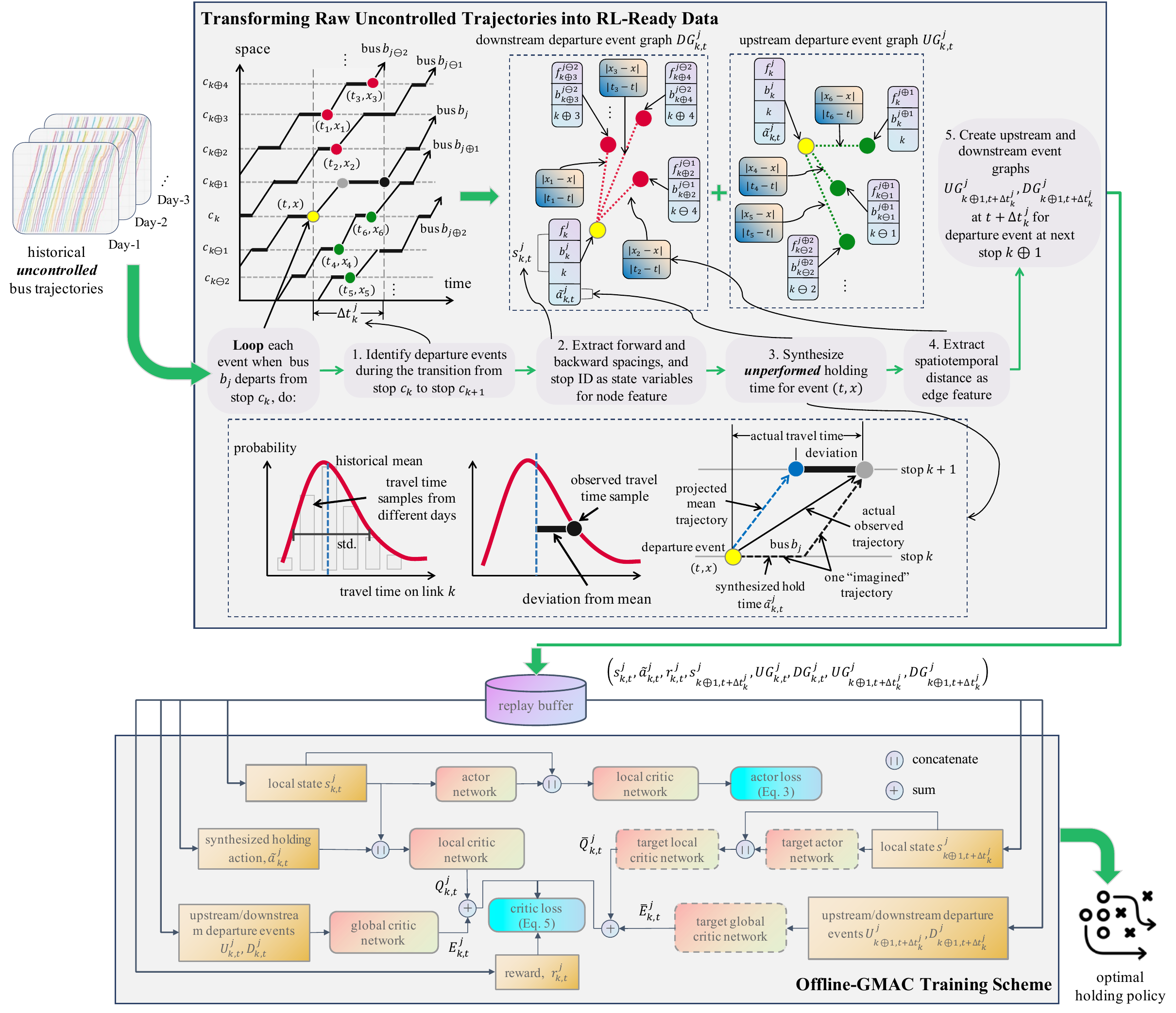}
    \caption{From uncontrolled data to optimized holding policies: Offline-GMAC data preparation and training.}
    \label{fig:overall_pipeline}
\end{figure*}

\subsubsection{Action}
The holding time of bus $b_j$ upon departure from stop $c_k$ at time $t$, denoted as $h_{k,t}^{j}$, is computed as:
\begin{equation}
	h_{k,t}^{j} = a_{k,t}^{j}\cdot W,
\end{equation}
where $a_{k,t}^{j} \in [0,1]$ is a continuous decision variable, representing the holding intensity determined by the agent; and $W$ denotes the maximum holding time pre-defined by the transit agency, indicating the tolerance for holding delay. When $a_{k,t}^{j}=0$, no holding is implemented.

\subsubsection{Local state transition dynamics}
Upon departing from stop $c_k$ at time $t$, bus $b_j$ observes a state $s^j_{k,t}$ and takes an action $a^j_{k,t}$. The bus then undergoes a random time interval $\Delta t^j_k$ before its next departure event at stop $c_{k\oplus 1}$ occurs, when it receives the subsequent observation $s^j_{k\oplus 1,t+\Delta t^j_k}$. The transition of the state can be represented by a conditional probability $p( s_{k\oplus 1,t+\Delta t^j_k}^j \mid s_{k,t}^j, h_{k,t}^j )$. The duration $\Delta t^j_k$ encompasses the holding time $h^j_{k,t}$, the stochastic travel time on link $l_k$, and the stochastic dwell time at stop $c_{k\oplus 1}$.

\subsubsection{Global shared reward}
We formulate the bus holding control as a cooperative multi-agent RL problem, where bus agents collaborate to enhance overall operational efficiency. To facilitate cooperation, the training of the actor is centralized, with all buses' local states and actions observed, and a shared global reward is defined. The reception of the reward for action $a_{k,t}^{j}$ is deferred until bus $b_j$ departs from stop $c_{k\oplus 1}$ at time $t+\Delta t^j_k$. The global shared reward is defined as follows:
\begin{equation}\label{eq:reward}
	r^j_{k,t} = -\sigma_{b_j\in B}[|x_{t+\Delta t^j_k}^{j\ominus 1} \ominus x_{t+\Delta t^j_k}^{j} | ] - w\cdot a_{k,t}^{j},
\end{equation}
where $\sigma[\cdot]$ represents the standard deviation operator across all buses, and $w$ is a balancing parameter that controls the tradeoff between system stability and holding intensity. The first term encourages buses to maintain even spacings, while the second term penalizes holding time.

\subsection{Policy and value functions}\label{sec:value_functions}
This subsection describes how the actor and critics are approximated by neural networks. 

\subsubsection{Actor network \texorpdfstring{$\pi_{\theta^j}$}{pi (theta j)}}
The actor network $\pi_{\theta^j}$, parametrized by $\theta^j$, receives the local state $s^j_{k,t}$ as input and produces a deterministic continuous action $a^j_{k,t}=\pi( s^j_{k,t};\theta^j )$ as output. Specifically, the forward spacing $f^j_{k,t}$ and backward spacing $b^j_{k,t}$ are fed into a Multilayer Perceptron (MLP) to generate a hidden vector of size $\kappa$. Additionally, the stop ID is encoded using an embedding layer to create another hidden vector of size $\kappa$. These two hidden vectors are then elementwise added and fed into another MLP to generate the final action. By incorporating the stop ID in this manner, we aim to capture diverse characteristics of bus stops and links, like heterogeneous passenger arrival rates and travel time distributions.
\begin{figure*}[!ht]
\centering
\subfloat[]{\includegraphics[width=0.235\textwidth]{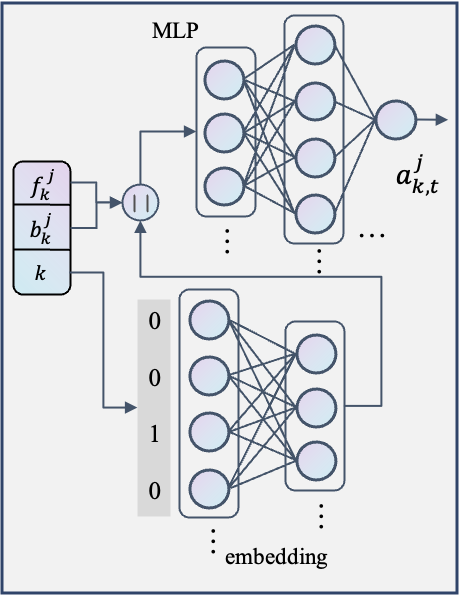}%
}
\hfil
\subfloat[]{\includegraphics[width=0.25\textwidth]{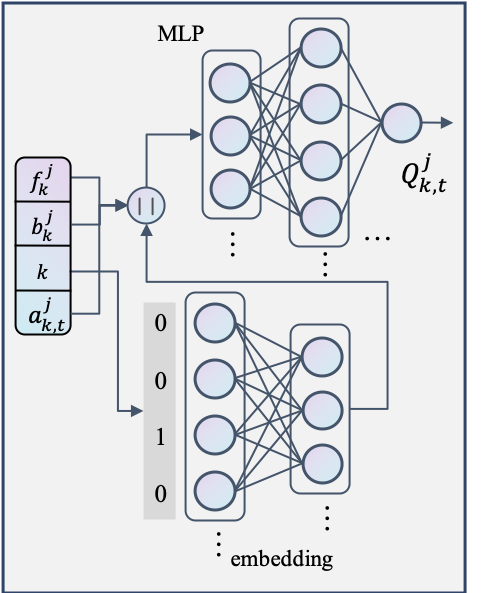}%
}
\hfil
\subfloat[]{\includegraphics[width=0.485\textwidth]{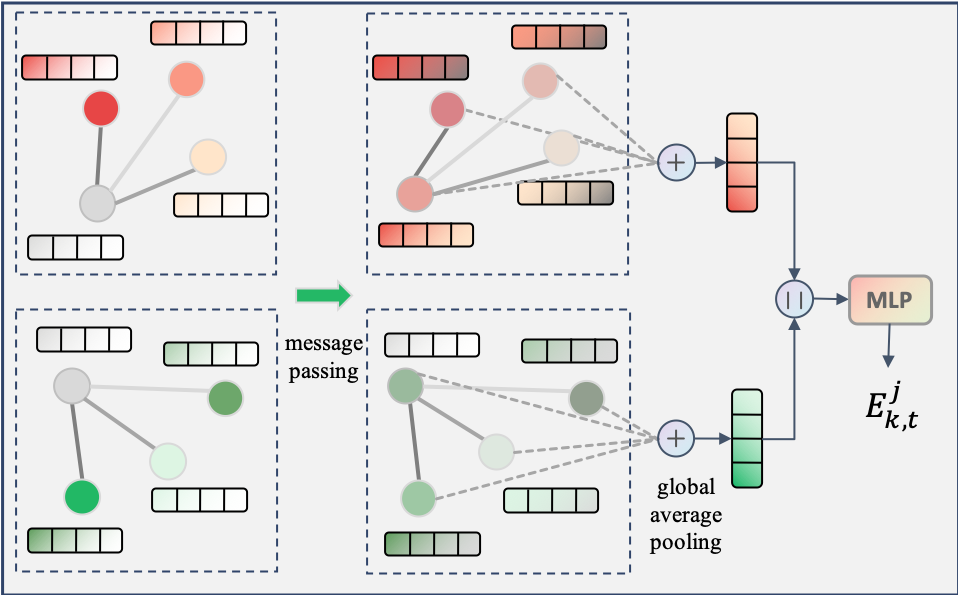}%
}
\caption{Neural network architectures. (a) Actor network. (b) Local critic network. (c) Global critic network.}
\label{fig:network_architecture}
\end{figure*}

\subsubsection{Local critic network \texorpdfstring{$Q_{\phi^j}$}{Q (phi j)}}
The local critic network $Q_{\phi^j}$, parametrized by $\phi^j$, takes the local state $s^j_{k,t}$ and action $a^j_{k,t}$ of agent $b_j$ as inputs and outputs a scalar value $Q( s^j_{k,t},a^j_{k,t};\phi^j )$, which estimates the expected return (i.e., the discounted cumulative sum of rewards) of conducting $a^j_{k,t}$ under $s^j_{k,t}$ and then following a policy $\pi_{\theta^j}$ for future holding decisions. The treatment of the local state part is similar to that of the actor network (see Fig.~\ref{fig:network_architecture}). The only difference is that the forward spacing $f^j_{k,t}$ and backward spacing $b^j_{k,t}$ are first concatenated with the chosen action $a^j_{k,t}$ and then fed into the MLP.

\subsubsection{Global event graph critic network \texorpdfstring{$E_{\varphi}$}{E (varphi)}}\label{sec:global_critic}
The local critic network only evaluates the individual agent's contribution to the global reward, disregarding the impacts of other agents. However, as depicted in Fig.~\ref{fig:ts_illustrate}, when bus $b_j$ at stop $c_k$ makes a holding decision at time $t^*$, its forward spacing is also influenced by the action taken by bus $b_{j\ominus 1}$. The effect on bus $b_j$'s forward spacing subsequently affects the received global reward defined in Eq.~\eqref{eq:reward}. To effectively account for the influence of actions from other agents on the current bus $b_j$, we introduce a global event graph critic as follows.\par

We first aggregate all the departure events and their corresponding actions that take place upstream and downstream of the current event during a single local transition period of bus $b_j$. Specifically, we construct the upstream event set, $U^j_{k,t}$, by gathering departure events occurring in the upstream half-circle of the loop corridor using stop $c_k$ as the reference point. The remaining departure events that occur in the downstream half-circle of the loop corridor are consolidated into the downstream event set, $D^j_{k,t}$. Formally, $U^j_{k,t} \equiv \{ (b_{j'},c_{k'},t')\vert b_{j'}\neq b_{j}, t'\in (t,t+\Delta t^j_k),k'\ominus k < \frac{N}{2}, \forall b_{j'}\in B \}$ and $D^j_{k,t} \equiv \{ (b_{j'},c_{k'},t')\vert b_{j'}\neq b_j, t'\in (t,t+\Delta t^j_k),k'\ominus k \geq \frac{N}{2}, \forall b_{j'}\in B \}$, where a tuple $(b_{j'},c_{k'},t')$ denotes an event of bus $b_{j'}$ departing from stop $c_{k'}$ at time $t'$ and $B$ is the bus fleet set. Note that this half-circle split provides complete fleet coverage with respect to bus $b_j$ and lets the GNN itself learn the interaction structure from data, rather than relying on a prescribed fraction that requires additional tuning.\par

Next, we construct an upstream event graph, $UG^j_{k,t}$, and a downstream event graph, $DG^j_{k,t}$, based on $U^j_{k,t}$ and $D^j_{k,t}$, respectively. In these graphs, each node represents a departure event and includes the local state observation and the corresponding action as the node attributes. The current departure event is connected to every other departure event in a star graph configuration. The edge attributes capture the spatial and temporal distances between pairs of departure events.\par

As a concrete example, consider Fig.~\ref{fig:ts_illustrate}: at the decision time $t^*$ of bus $b_j$ departing from stop $c_k$, the departure of the preceding bus $b_{j\ominus 1}$ from stop $c_{k\oplus 2}$ constitutes a node of the upstream event graph $UG^j_{k,t}$, with the local state observation of $b_{j\ominus 1}$ at that moment and its holding time at $c_{k\oplus 2}$ serving as the node attributes.

To collectively consider the impacts of all other actions, we employ a GNN to perform message passing in the downstream and upstream event graphs. This approach involves exchanging messages between nodes, allowing the GNN to effectively aggregate information regarding the effects of different actions on the global reward. Following the message passing, the two node embedding vectors of the current event $(b_j,c_k,t)$ from the upstream and downstream event graphs are element-wise summed. The resulting vector is then fed into an MLP to obtain the final global critic $E( UG^j_{k,t},DG^j_{k,t};\varphi )$ for the event $(b_j,c_k,t)$. It is noteworthy that various message-passing-based GNNs, such as graph convolutional networks (GCNs) and graph attention networks (GATs), can be employed in this framework \citep{jin2023spatio}. As a result, the method proposed in \citet{wang2021reducing} that uses the GAT to perform the message passing can be viewed as a special case of our proposed method.\par

In sum, the global event graph critic takes the upstream event graph $UG^j_{k,t}$ and the downstream event graph $DG^j_{k,t}$ as inputs, and outputs a scalar value, $E( UG^j_{k,t},DG^j_{k,t};\varphi )$, estimating the expected return contributed by the other agents' departure events that occur during one local transition period for the bus agent $b_j$.

\subsection{Training scheme}\label{sec:training_scheme}
Training of the actor $\pi_{\theta^j}$ is to adjust its parameters $\theta^j$ so as to maximize the value dictated by the local critic $Q_{\phi^j}$. The loss function of the actor is given as:
\begin{equation}\label{eq:ddpg_policy_objective}
	\mathcal{J}(\theta^j) = -{\mathbb{E}} {[ Q(s_{k,t}^j, \pi (s_{k,t}^j;{\theta^j} );{\phi^j} ) ]}.
\end{equation}\par

The actual expected return after $b_j$ implementing $a^j_{k,t}$ at time $t$ is approximated using the sum of the outputs from the local and global critics:
\begin{equation}\label{eq:sum_critics}
Z^j_{k,t}( s^{j}_{k,t},a^{j}_{k,t},UG^j_{k,t},DG^j_{k,t}; \phi^j, \varphi ) = Q( s^{j}_{k,t},a^{j}_{k,t};\phi^j ) + E( UG^j_{k,t},DG^j_{k,t};\varphi ).
\end{equation}\par

To train $Q_{\phi^j}$ and $E_{\varphi}$ concurrently, we adopt a bootstrap approach to minimize the mean-squared Bellman error function:
\begin{equation}\label{eq:ddpg_critic_objective}
\mathcal{L}(\phi^j,\varphi) = {\mathbb{E}} [\frac{1}{2} ( Z^j_{k,t}( s^{j}_{k,t},a^{j}_{k,t},UG^j_{k,t},DG^j_{k,t}; \phi,\varphi ) -r^{j}_{k,t}-\gamma \bar{Z}^j_{k\oplus 1,t+\Delta t^j_k} )^2 ],
\end{equation}
where $\gamma$ is the discounting factor and $\bar{Z}^j_{k\oplus 1,t+\Delta t^j_k}$ denotes the expected return when bus $b_j$ departs from the next stop $c_{k\oplus 1}$ at time $t+\Delta t^j_k$. $\bar{Z}^j_{k\oplus 1,t+\Delta t^j_k}$ is calculated by summing the outputs from the two critics at the next decision step:
\begin{equation}
	\bar{Z}^j_{k\oplus 1,t+\Delta t^j_k} = \bar{Q} ( s^{j}_{k\oplus 1,t+\Delta t^j_k},\bar{\pi} ( s^{j}_{k\oplus 1,t+\Delta t^j_k};\bar{\theta}^j ) ;\bar{\phi}^j) + \bar{E}( UG^j_{k\oplus 1,t+\Delta t^j_k},DG^j_{k\oplus 1,t+\Delta t^j_k};\bar{\varphi}^j ),
\end{equation}
where $\bar{Q}_{\bar{\phi}^j}$, $\bar{\pi}_{\bar{\theta}^j}$ and $\bar{E}_{\bar{\varphi}^j}$ are the target networks that trace the main networks $Q_{\phi^j}$, $\pi_{{\theta^j}}$ and $E_{\varphi}$ over the course of training. They help to stabilize the training process and prevent oscillations and divergence in the critic values. They are updated periodically (once per main network update) using a Polyak averaging technique:
\begin{equation}
	\{\bar{\theta}^j,\bar{\phi}^j,\bar{\varphi}\} \gets \tau \{\bar{\theta}^j,\bar{\phi}^j,\bar{\varphi}\} + (1-\tau)\{\theta^j,\phi^j,\varphi\},
\end{equation}
where $\tau$ is a hyperparameter between 0 and 1 that determines the weight assigned to the previous target network parameters.\par

During the training process, the expectations $\mathbb{E}[\cdot]$ in Eqs.~\eqref{eq:ddpg_policy_objective} and \eqref{eq:ddpg_critic_objective} are approximated using the Monte Carlo method with $\epsilon$-size batches of transition pairs sampled from the replay buffer $\mcd$, i.e., $\{ (s^j_{k,t},a^j_{k,t},r^j_{k,t},UG^j_{k,t},DG^j_{k,t},s^j_{k\oplus 1,t+\Delta t^j_k},UG^j_{k\oplus 1,t+\Delta t^j_k},DG^j_{k\oplus 1,t+\Delta t^j_k} ) \}_{i=1}^{\epsilon} \sim \mcd$.\par

These transition pairs within the replay buffer can be collected by any interacting behavior policy. Fig.~\ref{fig:overall_pipeline} illustrates the data preparation process when the behavior policy corresponds to the uncontrolled agent; i.e., when the training is conducted using uncontrolled bus data. (Details on this matter will be presented in Section \ref{sec:data_synthesis}.)\par

Given the sampled transition pairs, the following stochastic gradients of the loss functions \eqref{eq:ddpg_policy_objective} and \eqref{eq:ddpg_critic_objective} can be computed:
\begin{equation}
	\nabla_{\theta^j} \mathcal{J}(\theta^j) = -\frac{1}{\epsilon}\sum_{i=1}^{\epsilon} {[ \nabla_{\theta^j}\pi ( s_{k,t}^j;\theta^j ) } \cdot { \nabla_{\pi( s_{k,t}^j;{\theta^j} )}Q( s_{k,t}^j,\pi(s_{k,t}^j;{\theta^j});{\phi^j} ) ]},
\end{equation}
\begin{align}
	\nabla_{\phi^j} \mathcal{L}(\phi^j,\varphi) &= \frac{1}{\epsilon}\sum_{i=1}^{\epsilon}{[( Q( s^{j}_{k,t},a^{j}_{k,t};\phi^j ) + E( UG^j_{k,t},DG^j_{k,t};\varphi )} \nonumber \\ 
	&- {r^j_{k,t}-\gamma \bar{Z}^j_{k\oplus 1,t+\Delta t^j_k}  ) \cdot \nabla_{\phi^j}Q( s^{j}_{k,t},a^{j}_{k,t};\phi^j )]},
\end{align}
and
\begin{align}
	\nabla_\varphi \mathcal{L}(\phi^j,\varphi) &= \frac{1}{\epsilon}\sum_{i=1}^{\epsilon}{[( Q( s^{j}_{k,t},a^{j}_{k,t};\phi^j ) + E( UG^j_{k,t},DG^j_{k,t};\varphi )} \nonumber \\
	&-{r^j_{k,t}-\gamma \bar{Z}^j_{k\oplus 1,t+\Delta t^j_k}) \cdot \nabla_{\varphi}E( UG^j_{k,t},DG^j_{k,t};\varphi )]}.
\end{align}

The above gradients can then be used for updating the network parameters with a learning rate $\alpha$.\par

In our training scheme, we adopt the parameter-sharing technique commonly used in RL-based bus holding literature \citep{alesiani2018reinforcement,wang2020dynamic}. This means that all bus agents $b_j\in B$ share a common set of parameters for the actor, local critic, and global critic. This approach enables training the policy using the collective experience of all agents simultaneously, thereby reducing the data requirements for training. Note that the parameter sharing would not compromise the agents' ability to exhibit distinct behavior due to the different observations received by each agent. (We empirically validate this design choice through an ablation against a non-shared variant, where each bus maintains its own actor and local critic. Even when trained on $5\times$ more data, the non-shared variant still fails to match the performance of the shared architecture.) The training algorithm with parameter sharing is summarized in Algorithm \ref{algr:training}. Interested readers can refer to the source code for more details.

\begin{algorithm}[!ht]
\textbf{Input:} A total of $T$ historical transition pairs collected by any interacting behavior policies, $\{ (s^j_{k,t},a^j_{k,t},r^j_{k,t}$, $UG^j_{k,t},DG^j_{k,t},s^j_{k\oplus 1,t+\Delta t^j_k},UG^j_{k\oplus 1,t+\Delta t^j_k},DG^j_{k\oplus 1,t+\Delta t^j_k} )$ $,\forall b_j\in B \}$, number of training epoch $E$, $\epsilon$, $\tau$, and $\alpha$.\\
Combine the transition pairs from all agents regardless of agent index and store them in a shared replay buffer $\mcd$: $\mcd\equiv \{ (s^i_{k,t},a^i_{k,t},r^i_{k,t},UG^i_{k,t},DG^i_{k,t}$ $,s^i_{k\oplus 1,t+\Delta t^j_k},UG^i_{k\oplus 1,t+\Delta t^j_k}$ $,DG^i_{k\oplus 1,t+\Delta t^j_k}) \}_{i=1}^{T}$, where $i$ indexes the transition tuple.\\
Initialize main actor network $\pi_{\theta}$, main local critic network $Q_{\phi}$, and main global event critic network $E_{\varphi}$, shared by all the bus agents $b_j\in B$. \\
Initialize target actor network $\bar{\pi}_{\bar{\theta}}$, target local critic network $\bar{Q}_{\bar{\phi}}$, and global event critic network $\bar{E}_{\bar{\varphi}}$ and copy the parameters from the main networks.\\
\For{epoch=1 to $E$}{
	Divide the transition pairs in $\mcd$ into $K=\lfloor \frac{T}{\epsilon} \rfloor $ mini-batches of size $\epsilon$ and shuffle.\\
	\For{mini-batch=1 to $K$ }{
		\For{transition pair $i$=1 to $\epsilon$}{
		Calculate local and global critic values at stop $k$ using main critic networks and sum them up:\\
		$Z^i_{k,t}= Q( s^{i}_{k,t},a^{i}_{k,t};\phi ) + E( UG^i_{k,t},DG^i_{k,t};\varphi).$\\
		Calculate local and global critic values at the next stop $k\oplus 1$ using target critic networks and sum them up: \\
		$\bar{Z}^i_{k\oplus 1,t+\Delta t^j_k} = \bar{Q} ( s^{i}_{k\oplus 1,t+\Delta t^j_k},\bar{\pi} ( s^{i}_{k\oplus 1,t+\Delta t^j_k};\bar{\theta} ) ;\bar{\phi}) + \bar{E}( UG^i_{k\oplus 1,t+\Delta t^j_k},DG^i_{k\oplus 1,t+\Delta t^j_k};\bar{\varphi} )$. \\
		Calculate the mean-squared Bellman error for the $i$-th transition:\\
		$L^i{(\phi,\varphi)} = \frac{1}{2}{( Z^i_{k,t} -r^{i}_{k,t}-\gamma \bar{Z}^i_{k\oplus 1,t+\Delta t^j_k} )^2 } $. \\
		Calculate actor loss for the $i$-th transition:\\
		$\mathcal{J}(\theta)^i = -Q(s_{k,t}^i, \pi (s_{k,t}^i;{\theta} );{\phi} ) $.
		}
	Form batch critic loss as $\mathcal{L}(\phi,\varphi) = \frac{1}{\epsilon}\sum_{i=1}^{\epsilon}{L^i{(\phi,\varphi)}}$ and perform a gradient descent step on $\phi$ and $\varphi$ with $\alpha$.\\
	Form batch actor loss as $\mathcal{J}(\theta) = \frac{1}{\epsilon}\sum_{i=1}^{\epsilon}{\mathcal{J}^i(\theta)}$ and perform a gradient descent step on $\theta$ with $\alpha$.\\
	Update target networks by polyak averaging:\\
	$\{\bar{\theta},\bar{\phi},\bar{\varphi}\} \gets \tau \{\bar{\theta},\bar{\phi},\bar{\varphi}\} + (1-\tau)\{\theta,\phi,\varphi\}$.
	}
}
\textbf{Output:} Trained actor network $\pi_{\theta}^{*}$, local critic network $Q_{\phi}^{*}$, and global event critic network $E_{\varphi}^{*}$.

\caption{Offline-GMAC training process with parameter sharing}\label{algr:training}
\end{algorithm}

In our experiments, we fix the number of training epochs at $E=20$, since the offline dataset is static and training beyond this budget does not noticeably change policy quality.

\subsection{Training data synthesis}\label{sec:data_synthesis}

To train an offline RL policy, we need to use comprehensive and diverse data sets collected under various holding controls. Based on our RL element design (Section \ref{sec:rl_elements}), the essential components of the state $s^j_{k,t}$ (comprising bus spacings) and the reward $r^j_{k,t}$ (which measures spacing deviation) can both be derived from the instantaneous bus locations extracted from historical GPS trajectories. Unfortunately, most available bus GPS trajectory datasets do not contain holding control information. In RL terminology, these data sets provide information about the transition dynamics $\widetilde{\mcp}^j$, where the action space $\mca^j$ only consists of actions $a_{k,t}^j=0$ for all $j$, $k$, and $t$. In other words, they represent the conditional probability $p( s_{k\oplus 1,t+\Delta t^j_k}^j \mid s_{k,t}^j,a_{k,t}^j=0 )$. Essentially, this corresponds to a Markov Reward Process (MRP) with the interacting policy being no-control.\par

Holding-controlled data are rare and expensive to obtain since implementing holding actions in reality is costly. Given this, we wonder whether the historical uncontrolled bus trajectories can be used to ``synthesize'' holding-controlled data for the training purpose. If the answer is affirmative, then we will be able to develop an effective holding policy \textit{offline} solely utilizing historical GPS data, even before any control measures are implemented. This would eliminate the adverse effects caused by ``bad'' holding actions during the initial training phase of an online RL policy and also alleviate the challenges associated with data acquisition for offline RL training.\par

To construct such a synthesis scheme, note that holding control at a stop essentially improves system stability by delaying the buses running ahead. On the other hand, a fast bus will also be delayed if its travel time in the next link becomes longer (e.g., caused by a red signal or heavy traffic). Hence, a longer link travel time can be regarded as having the same effect as a hold on the bus's arrival time at the next stop (see Fig.~\ref{fig:overall_pipeline}). (Note that this synthesis makes no assumption about whether passengers continue to board during a hold; the trained policy yields comparable performance whether or not boarding-during-hold is enabled in the operating environment, as we will show in Section~\ref{sec:pareto_comparison}.) To act on this insight, we define $Y_{k,t}^j$ as the random variable for the historical travel time on link $l_k$ for bus $b_j$ departing stop $c_k$ at time $t$, which can be extracted from GPS trajectories. $y_{k,t}^j$ denotes a specific, observed realization of this variable extracted from our historical uncontrolled GPS dataset. We then define $Y_k$ as the random variable representing the aggregate empirical distribution formed by all of these historical realizations $\{y_{k,t}^j\}$ for a given link $l_k$. Building on these observed travel times, we propose a scheme to synthesize the \textit{unperformed} holding action, as outlined in Algorithm \ref{algr:infer_a}.\par

\begin{algorithm}[!ht]
\textbf{Input:} Historical uncontrolled bus trajectories for all buses $b_j\in B$; Maximum holding time $W$; Tunable hyperparameter $\eta$. \\
Extract experienced link travel times $\{ y_{k, t}^j \}$ for bus $b_j$ on link $l_k$ after it departs from stop $c_k$ at time $t$.\\
\For{all $l_k\in L$}{
	Aggregate $y_{k, t}^j$ of each bus $b_j\in B$ departing from stop $c_k$ at time $t$ and form a distribution $Y_k$ with mean $\mathbb{E}[Y_k]$ and standard deviation $\sigma[Y_k]$.\\
	Calculate a ``regulator'' $\frac{W}{\eta\cdot\sigma[Y_k]}$ with a hyperparameter $\eta$.\\
	\For{each transition $s_{k,t}^j \mapsto s_{k\oplus 1,t+\Delta t^j_k}^j$}{
	$\tilde{h}_{k,t}^j \gets \frac{W}{\eta\cdot\sigma[Y_k]}\cdot \max\{ 0, y_{k, t}^j-\mathbb{E}[Y_k] \}$ \\
	$\tilde{a}_{k,t}^j \gets \min\{ \tilde{h}_{k,t}^j, W\}/W$ .\\
	Associate $\tilde{h}_{k,t}^j$ and $\tilde{a}_{k,t}^j$ with the transition $s_{k,t}^j \mapsto s_{k\oplus 1,t+\Delta t^j_k}^j$.
	}
}
\textbf{Output:} Synthesized unperformed $\tilde{a}_{k,t}^j$ for all $b_j$ and $c_k$.

\caption{Synthesizing holding actions from uncontrolled bus trajectories}\label{algr:infer_a}
\end{algorithm}

The regulator $\frac{W}{\eta\cdot\sigma[Y_k]}$ in lines 5 and 7 of Algorithm \ref{algr:infer_a} is utilized to adjust the synthesized holding action based on travel time variability. A smaller standard deviation $\sigma[Y_k]$ indicates that the realized travel times $y_{k, t}^j$'s deviate slightly from the mean $\mathbb{E}[Y_k]$. In such cases, a larger regulator is employed to amplify the action. Conversely, a larger standard deviation is paired with a smaller regulator, which serves to moderate the action. The hyperparameter $\eta$ in the regulator serves as a tuning knob to control the responsiveness of the regulator to the standard deviation $\sigma[Y_k]$. The theoretical motivation for the introduction of $\eta$ is provided in Proposition \ref{prop:distribution_shift_bound} (Section \ref{sec:theoretical_justification_data_synthesis}).\par

\subsection{Theoretical justifications for the data synthesis scheme}\label{sec:theoretical_justification_data_synthesis}

A central concern of offline RL is extrapolation error \citep{levine2020offline}, which stems from state--action coverage gaps and transition distribution shift; in our setting, the former is naturally mitigated by the stochasticity of link travel times, while the latter, induced by our synthesis rule, is the focus of this section.\par

Our data synthesis scheme involves a distributional shift between the observed next states in historical uncontrolled data and the desired, counterfactual next states that would have resulted from applying the synthesized holds. The effectiveness of our proposed scheme relies on how closely these two distributions align. To theoretically quantify this closeness, we first formally define the two distributions. (We will also numerically verify this alignment in Section~\ref{sec:closeness_verification}.)

\begin{definition}
	(The desired distribution, $p_{des}$) For any given state $s_{k,t}^j$ and any holding time $h\geq 0$, the desired distribution, $p_{des}(\cdot | s_{k,t}^j,h)$, is the virtual probability distribution over the next state, $s_{k\oplus 1,t+\Delta t^j_k}^j$, that would arise from applying the holding action $h_{k,t}^j = h$.
\end{definition}

In this process, the system evolves under a set of environmental random variables. This set encompasses all sources of stochasticity in the system, such as passenger arrival processes at all stops and the traffic conditions on all links. A bus $b_j$ would, for instance, experience a random link travel time $T_{k,t}^j$ on link $l_k$ when departing from stop $c_k$. This notation is for the deployment-environment travel time and to distinguish it from the $Y_{k,t}^j$, which is derived from our observed dataset.

\begin{definition}
	(The observed distribution, $p_{obs}$) Let the synthesized holding time, $H_s$, be the random variable defined by the synthesis rule in Algorithm~\ref{algr:infer_a}: $H_s \equiv \frac{W}{\eta\cdot\sigma[Y_k]}\cdot \max\{ 0, Y_{k,t}^j-\mathbb{E}[Y_k] \}$. For any given state $s_{k,t}^j$ and any holding time $h\geq 0$, the observed distribution, $p_{obs}(\cdot | s_{k,t}^j,h)$, is the distribution of historical outcomes that occurred when the actual hold was zero, conditioned on the event that $H_s$ equals the value $h$.
\end{definition}

Unfortunately, an exact analysis of the closeness between these two distributions proves intractable due to the complex, coupled nature of bus system dynamics. The state transition for a single bus agent between two consecutive stops depends not only on its own travel time and holding time but also on the unobserved actions and stochastic movements of the preceding and succeeding buses during the local transition period (see Fig.~\ref{fig:ts_illustrate}). Furthermore, stochastic processes such as passenger arrivals at the upcoming stop influence the transition. To make the analysis tractable, we separate the distribution shift caused by our synthesis scheme alone from the corresponding effect that may be amplified or reduced by the stochastic, coupled dynamics of the bus system.\par

To understand why our synthesis scheme causes distributional shift, we first recall that the scheme (Algorithm \ref{algr:infer_a}) rests on the equivalence between a deliberate holding action and a stochastic delay in link travel time. However, this equivalence is valid only for positive deviations from the expected travel time. We can map such delays to equivalent holding actions, but we cannot map early arrivals to ``negative holds.'' The $\max\{0,\cdot\}$ truncation in our synthesis rule reflects this practical constraint. This asymmetry fundamentally introduces distributional shift in bus arrival times at the next stop.\par

Therefore, we organize the analysis in this section as follows. We begin in Section~\ref{sec:components_of_analysis} by defining the desired arrival time (what would have happened with a hold) and the observed arrival time (what we infer from historical data), an underlying dynamics map $G$ that encapsulates all unmodeled complexities, and a mathematical metric ($\mathcal{W}_1$) needed to measure the distance between distributions. Then, in Section~\ref{sec:arrival_distribution_shift}, we bound the distance between the two arrival distributions. Finally, in Section~\ref{sec:transition_shift_bound}, we demonstrate how this input bound propagates through the dynamics map $G$ to yield a final bound on the output state transition distributions (i.e., $p_{des}$ and $p_{obs}$).

\subsubsection{Components of the analysis: arrival times, map \texorpdfstring{$G$}{G}, and Wasserstein metric}\label{sec:components_of_analysis}

We begin by defining the two arrival time distributions. The desired arrival time of bus $b_j$ at stop $c_{k\oplus 1}$, $\hat{t}_{\text{arr},k\oplus 1}^{j}$, is a random variable calculated as $\hat{t}_{\text{arr},k\oplus 1}^{j} \equiv t_{\text{dep},k}^{j}+h+T_{k,t}^j$, where $t_{\text{dep},k}^{j}$ is the departure time of bus $b_j$ from stop $c_k$. The observed arrival time of bus $b_j$ at stop $c_{k\oplus 1}$, $\bar{t}_{\text{arr},k\oplus 1}^{j}$, is a random variable given as $\bar{t}_{\text{arr},k\oplus 1}^{j}\equiv t_{\text{dep},k}^{j}+0+Y_{k,t}^j$, where $Y_{k,t}^j$ is the historical travel time, conditioned on the event $H_s=h$.\par

Next, we define the underlying system dynamics map $G$, which encapsulates all unmodeled complexities, such as the passenger service process at the next stop and the coupled dynamics of neighboring buses in the fleet. The map $G$ takes two arguments: a state $s_{k,t}^j$ and a deterministic arrival time at the next stop, $t_a$. It outputs the resulting conditional probability distribution over the next state, $s_{k\oplus 1,t+\Delta t^j_k}^j$:
\begin{equation}
	G(s_{k,t}^j,t) \equiv p(s_{k\oplus 1,t+\Delta t^j_k}^j | s_{k,t}^j, t_{\text{arr},k\oplus 1}^{j}=t),
\end{equation}
where $t_{\text{arr},k\oplus 1}^{j}$ is the arrival time of bus $b_j$ at stop $k\oplus 1$. We assume that this underlying dynamics map $G$ is the same for both the desired and observed processes. Consequently, $p_{des}$ and $p_{obs}$ are the pushforward measures of the arrival time distributions $P(\hat{t}_{\text{arr},k\oplus 1}^{j})$ and $P(\bar{t}_{\text{arr},k\oplus 1}^{j})$ through the common map $G$.

Finally, to quantify the distance between these distributions, we introduce our mathematical metric. We model the state space as a complete and separable metric space $(\mathcal{X}, d)$. Let $\mathcal{P}_1(\mathcal{X})$ denote the set of all probability measures on $\mathcal{X}$ with finite first moment, i.e., measures $\mu$ that satisfy $\int_{\mathcal{X}} d(x_0, x) \, d\mu(x) < \infty$ for some $x_0 \in \mathcal{X}$. We measure the closeness between two measures $\mu, \nu \in \mathcal{P}_1(\mathcal{X})$ via the Wasserstein-1 distance, which admits the following well-known dual formulation \citep{villani2008optimal}:
\begin{equation}\label{eq:dual_formulation}
	\mathcal{W}_1(\mu, \nu) = \sup_{\|\psi\|_{\text{Lip}} \le 1} \left\{  \int_{\mathcal{X}} \psi(x) \, d\mu(x) - \int_{\mathcal{X}} \psi(x) \, d\nu(x)  \right\},
\end{equation}
where the supremum ranges over all 1-Lipschitz functions $\psi: \mathcal{X} \to \mathbb{R}$.\par

\subsubsection{Arrival time distribution shifts and Lipschitz continuity of \texorpdfstring{$G$}{G}}\label{sec:arrival_distribution_shift}
With all components defined, we now proceed to execute our analysis.

\begin{lemma}\label{lemma:distance_of_arrivals}
	Assume the actual travel time $T_{k,t}^j$ and the historical travel time $Y_{k,t}^j$ are independent and identically distributed random variables with mean $\mu_k$ and standard deviation $\sigma_k$. For each state $s_{k,t}^j$ and any holding time $h\geq 0$, the 1-Wasserstein distance between the desired and observed arrival time distributions is bounded as follows:
	\begin{equation}
		\mathcal{W}_1(P(\hat{t}_{\text{arr},k\oplus 1}^{j}), P(\bar{t}_{\text{arr},k\oplus 1}^{j})) = \mathcal{O}(|W/\eta - \sigma_k| + \sigma_k).
	\end{equation}
	\begin{proof}
		Please see Appendix~\ref{apdx:lemma_distance_of_arrivals_proof}.
	\end{proof}
	
\end{lemma}
Next, we establish the tools required to propagate this input bound through the dynamics map $G$. This propagation requires a key assumption about the regularity of the map.

\begin{assumption}\label{assumption:L_continuity}
	For each state $s_{k,t}^j$, the map $G(s_{k,t}^j,\cdot )$ is Lipschitz continuous in its second argument. That is, there exists a constant $L_s$ (which may depend on the state $s$) such that for all arrival times $t_a$ and $t_b$:
	\begin{equation}
		\mathcal{W}_1(G(s_{k,t}^j,t_a), G(s_{k,t}^j,t_b)) \leq L_s |t_a - t_b|.
	\end{equation}
\begin{remark}
	This assumption formalizes that the physical bus system is not overly sensitive: a small change in a bus's arrival time results in a proportionally small, bounded shift in the output distribution of system states. This assumption is generally satisfied in bus systems, as the $L_s$ can take large values as needed.
\end{remark}
\end{assumption}

This assumption bounds the sensitivity of the map for deterministic inputs. The following theorem extends scalar inputs to distributions (i.e., our $P(\hat{t})$ and $P(\bar{t})$), and shows that the Lipschitz property is preserved under the pushforward operation.

\begin{theorem}\label{theorem:push_forward}
	Let $(\mathcal{X}, d_\mathcal{X})$ and $(\mathcal{Y}, d_\mathcal{Y})$ be complete and separable metric spaces. Let $K: \mathcal{Y} \to \mathcal{P}_1(\mathcal{X})$ be a stochastic kernel that is $L_s$-Lipschitz, i.e., for any $y_1, y_2 \in \mathcal{Y}$:
	\begin{equation}\label{eq:kernel_assumption}
		\mathcal{W}_1(K(y_1), K(y_2)) \leq L_s \cdot d_\mathcal{Y}(y_1, y_2).
	\end{equation}
	Let $K_\#$ be the operator that maps an input distribution $P \in \mathcal{P}_1(\mathcal{Y})$ to the corresponding mixture distribution $K_\#P \in \mathcal{P}_1(\mathcal{X})$, defined by:
	\begin{equation}\label{eq:mix_definition}
		(K_\#P)(A) = \int_{\mathcal{Y}} K(y)(A) \, dP(y) \quad \text{for any measurable } A \subseteq \mathcal{X}.
	\end{equation}
	Then, this mixing operator $K_\#: \mathcal{P}_1(\mathcal{Y}) \to \mathcal{P}_1(\mathcal{X})$ is also $L_s$-Lipschitz:
	\begin{equation}
		\mathcal{W}_1(K_\#P_1, K_\#P_2) \leq L_s \cdot \mathcal{W}_1(P_1, P_2) \quad \text{for all } P_1, P_2 \in \mathcal{P}_1(\mathcal{Y}).
	\end{equation}
	\begin{proof}
		Please see Appendix~\ref{apdx:theorem_push_forward}.
	\end{proof}
\end{theorem}

\subsubsection{State transition distribution shift bounds}\label{sec:transition_shift_bound}

We combine the input bound from Lemma~\ref{lemma:distance_of_arrivals} with the propagation rule from Theorem~\ref{theorem:push_forward} to arrive at our main result:
\begin{proposition}\label{prop:distribution_shift_bound} Assume the actual travel time $T_{k,t}^j$ and the historical travel time $Y_{k,t}^j$ are independent and identically distributed random variables with mean $\mu_k$ and standard deviation $\sigma_k$. Under Assumption~\ref{assumption:L_continuity}, for any given state $s_{k,t}^j$ and any holding time $h\geq 0$, the 1-Wasserstein distance between the desired state distribution, $p_{des}$, and the observed state distribution, $p_{obs}$, is bounded as follows:
	\begin{equation}
		\mathcal{W}_1(p_{des}(\cdot | s_{k,t}^j, h), p_{obs}(\cdot | s_{k,t}^j, h)) = \mathcal{O}(L_s ( |W/\eta - \sigma_k| + \sigma_k )).
	\end{equation}
\begin{proof}
	The state distributions $p_{des}$ and $p_{obs}$ are the pushforwards of their respective arrival time distributions through the map $G(s_{k,t}^j,\cdot)$. By Theorem~\ref{theorem:push_forward}, the distance between these output distributions is bounded by $L_s$ times the distance between the input distributions:
	\begin{equation}
		\mathcal{W}_1(p_{des}(\cdot | s_{k,t}^j, h), p_{obs}(\cdot | s_{k,t}^j, h))\leq L_s\cdot \mathcal{W}_1(P(\hat{t}_{\text{arr},k\oplus 1}^{j}), P(\bar{t}_{\text{arr},k\oplus 1}^{j})).
	\end{equation}\par
	
	By Lemma~\ref{lemma:distance_of_arrivals}, the distance between the input arrival time distributions is $\mathcal{O}(|W/\eta - \sigma_k| + \sigma_k)$.\par
	
	Substituting the result of Lemma~\ref{lemma:distance_of_arrivals} into the inequality yields the final bound:
	\begin{equation}
		\mathcal{W}_1(p_{des}(\cdot | s_{k,t}^j, h), p_{obs}(\cdot | s_{k,t}^j, h))\leq L_s\cdot \mathcal{O}(|W/\eta - \sigma_k| + \sigma_k) = \mathcal{O}(L_s ( |W/\eta - \sigma_k| + \sigma_k )).
	\end{equation}
\end{proof}
\end{proposition}

\begin{remark}
This bound establishes that the distance between the desired distribution $p_{des}$ (the one we seek to learn) and the observed distribution $p_{obs}$ (the one available in our historical data) scales with $L_s$ multiplied by the bound from Lemma~\ref{lemma:distance_of_arrivals}. The bound disentangles two different factors: the system's intrinsic sensitivity to arrival time perturbations (captured by $L_s$) and the approximation error inherent in our synthesis scheme (captured by $\mathcal{O}( |W/\eta - \sigma_k| + \sigma_k )$).\par

The bound also reveals a pathway for mitigating the distributional shift. When we calibrate the regulator parameter $\eta$ such that $\eta \approx W/\sigma_k$, the $|W/\eta - \sigma_k|$ vanishes, which reduces the approximation error to its minimal achievable level.\par

The dependence on $\sigma_k$ further exposes a tradeoff in practice: a small travel time variability reduces the approximation error between synthesized and true controlled dynamics, but may also limit the diversity of training data; conversely, large values of $\sigma_k$ enable richer state-action coverage in the training dataset, yet this enrichment comes at the cost of amplified distributional shift.
\end{remark}

\section{Experiments}\label{sec:experiments}
To assess the performance of the proposed Offline-GMAC, we develop a simulator that captures key operational characteristics of a busy real-world bus route. It is important to note that this simulator is solely employed for \textit{assessment purposes}, and our offline method does not rely on it for policy training. This is different from previous online RL studies, where a simulator frequently functions as an environment for both training and testing. The real-world implementation of our method simply necessitates access to historical bus trajectories, whether involving holding control or not.\par

Section \ref{sec:simulator} details the real-world bus route data and simulator development. Section \ref{sec:comparison_strategies} introduces the strategies to be compared with our Offline-GMAC. Section \ref{sec:comp_on_off} compares the Offline-GMAC trained by uncontrolled data with its Online-GMAC counterpart, focusing on the initial implementation stage. Section \ref{sec:pareto_comparison} benchmarks our Offline-GMAC against existing model-based and model-free strategies over the medium to long run. In this section, the Pareto frontiers of these strategies under various hyperparameters are compared. Finally, Section \ref{sec:closeness_verification} provides empirical justifications for the proposed data synthesis scheme.

\subsection{Simulator for a real-world bus route}\label{sec:simulator}
The loop bus route No. 34, located in Chengdu, China, has been selected for the purpose of this study. It circulates along the First Ring Road of Chengdu, encompassing a total of 43 stops. To focus on bus bunching, which tends to occur more frequently during peak periods \citep{daganzo2009headway}, we choose to study the morning peak period from 7:30 to 9:30. We define this two-hour window as an \textit{episode}.\par

The simulator, developed using Python with a time-based logic, consists of two modules: the bus stop module and the inter-stop link module. At each time step $\delta t$, the simulator executes the two modules sequentially. Key elements of the simulator are introduced below. (For more details regarding the simulator, please refer to our online code repository.)

\subsubsection{Operations at stops}
Passenger arrivals at stop $c_k$ follow a Poisson process with a rate of $\lambda_k$. Each stop $c_k$ is associated with a set of downstream stops denoted by $\Omega_k$, representing the destinations for passengers boarding at $c_k$. A probability $\rho_{ki}$ is assigned to each specific stop $c_i\in \Omega_k$, ensuring that $\sum_{c_i\in \Omega_k}\rho_{ki} = 1$. This implies that the demand rate for passengers traveling from stop $c_k$ to stop $c_i$ is given by $\lambda_k\cdot\rho_{ki}$.\par
Buses operating on route No.~34 are equipped with two doors: the front door for passenger boarding and the rear door for alighting. At stop $c_k$, passengers board a bus at a rate of $\omega_k$ (passengers/sec) and alight at a rate of $\vartheta_k$ (passengers/sec). The bus dwell time at a stop is dictated by the maximum of the boarding and alighting times. Each bus on route No. 34 can carry a maximum of 100 passengers.\par
Due to the lack of real passenger OD data, estimating the $\rho_{kj}$'s is not feasible. We therefore only calibrated the $\lambda_k$'s while adopting the $\rho_{kj}$ values from Table 3 of \citet{he2022dynamic}. Fig.~\ref{fig:line34_parameters} depicts the calibrated $\lambda_k$ values by the gray bars where the calibration procedure is delegated to Appendix~\ref{apdx:arrival_rate_estimation}. Recall that the simulation serves the purpose of validation and demonstration, and thus moderate deviations from the reality will not hinder these purposes.\par

\begin{figure}[!ht]
    \centering
    \includegraphics[scale=0.45]{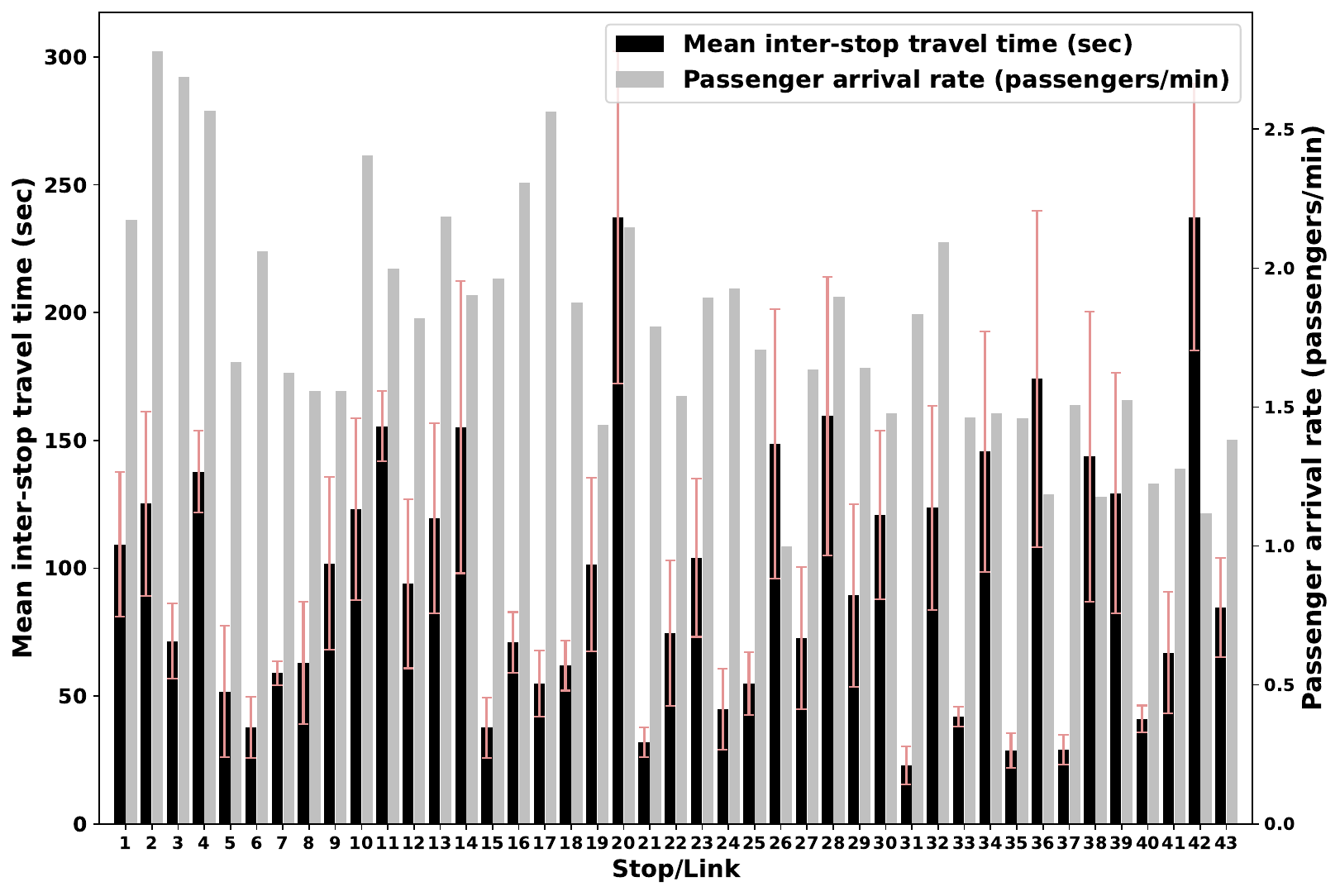}
    \caption{Estimated parameters for bus route No. 34.}\label{fig:line34_parameters}
\end{figure}\par

\subsubsection{Operations on inter-stop links}
The travel times on link $l_k$ connecting stops $c_k$ and $c_{k\oplus 1}$ are modeled as random variables following specific distributions, which are fitted to the GPS data using maximum likelihood estimation. We found that the lognormal distribution fits the data best among all the distributions tested. The black bars and the red error bars in Fig.~\ref{fig:line34_parameters} mark the means and standard deviations of the fitted inter-stop travel time distributions.
 
\subsection{Strategies for comparison}\label{sec:comparison_strategies}
 To comprehensively assess our proposed Offline-GMAC, we simulate and compare the following alternative strategies.
\begin{enumerate}
	\item \label{enum:alternative_on_gmc} \textbf{Online-GMAC.} Online RL methods were claimed to outperform non-RL-based approaches \citep{chen2016real,alesiani2018reinforcement,wang2020dynamic,wang2021reducing,he2022dynamic}. Thus, we select the online variant of our Offline-GMAC as a benchmark for comparison. In this online method, experiences (i.e., transition pairs) are continuously added to the replay buffer of size $|\mcd|$, employing the most recently learned policy (while under Offline-GMAC, a trained policy will remain unchanged during implementation). Additionally, to encourage exploration, a Gaussian noise with mean 0 and standard deviation $\upsilon$ is added to the output of the actor network. This exploratory noise empowers the online algorithm with the ability to search for superior solutions \citep{lillicrap2015continuous}. The  value of $\upsilon$ decays over episodes as $\upsilon(e) = \upsilon_0\cdot \xi^{e}$, where $\upsilon_0$ is the initial value, $\xi$ is the decay rate, and $e$ is the episode index. This approach of adaptive noise tuning strikes a balance between early exploration and later exploitation. Note that the same simulator is used for both training and testing Online-GMAC.
	\item \label{enum:alternative_off_gmc} \textbf{Offline-GMAC trained by uncontrolled data.} We use a dataset comprising 50 episodes of uncontrolled trajectories generated by the simulator to train an optimal holding policy. The ``holds'' are synthesized using Algorithm \ref{algr:infer_a} in Section \ref{sec:data_synthesis}. Fig.~\ref{fig:overall_pipeline} shows the data preparation process during which uncontrolled trajectories are transformed into transition pairs with synthesized holds.
	\item \label{enum:alternative_fhc} \textbf{Modified forward headway-based control (FHC) strategy, hereafter FHC for brevity.} This strategy is adapted from the adaptive control policy proposed by \citet{daganzo2009headway}, which incorporates a slack term to ensure non-negative holding times. We found that the slack term slightly degrades performance on the loop corridor, so we remove it and retain only a single control parameter $\zeta$ (corresponding to Daganzo's $\alpha$), with a max operator to ensure non-negativity. The original FHC strategy has been frequently utilized as a benchmark representing the performance of model-based control methods, as exemplified in \citet{wang2020dynamic}. Under the FHC strategy, the holding time is determined by:
\begin{equation}\label{eq:Daganzo2009}
	h^j_{k,t} = \max{ \{0, \zeta( \bar{H}-(t_{\text{dep},k}^j-t_{\text{dep},k}^{j\ominus 1} )) \}},
\end{equation}
where $\zeta$ is a control parameter adjusting holding intensity, $\bar{H}$ is the equilibrium headway for the loop corridor and $t_{\text{dep},k}^j$ is the time when bus $b_j$ departs from stop $c_k$. Determining $\bar{H}$ for the loop corridor is a little intricate due to its reciprocal influence with the holding times $h^j_{k,t}$'s. We delegate its calculation to Appendix~\ref{apdx:deter_psi}.
	\item \label{enum:alternative_off_gmc_fhc} \textbf{Offline-GMAC trained by FHC data.}
	We use a second dataset consisting of 50 episodes of trajectories, this time generated by the simulator under the FHC strategy, to train another optimal holding policy. This approach mimics having expert policies implemented in reality for generating the training data---a common practice in offline RL algorithms addressing diverse control problems in other domains \citep{agarwal2020optimistic}.
\end{enumerate}

\subsection{Comparisons between Online-GMAC and Offline-GMAC trained by uncontrolled data during initial implementation}\label{sec:comp_on_off}

\subsubsection{Hyperparameter selections}\label{sec:hyperparameter_selection}
Both the Online- and Offline-GMAC methods employ a single-layer GCN \citep{kipf2016semi} with a hidden size of 16 for conducting message passing. (Our findings indicate that the GCN slightly outperforms the GAT \citep{velivckovic2017graph, wang2021reducing}, but the performance gap tends to be small. Furthermore, variations in hidden sizes have minimal impacts on the results.) The $w$ and $\gamma$ are set to 0.005 and 0.9, respectively. Results under different $w$ will be presented in Section \ref{sec:pareto_comparison}. \par
To ensure fair comparisons, we conducted a grid search over a range of hyperparameter configurations for the Online-GMAC method. The search space was defined by the Cartesian product of the following parameter value sets: $\alpha \in \{5\mathrm{e}-4,1\mathrm{e}-4,5\mathrm{e}-3\}$, $\epsilon \in \{32, 64, 96, 128\}$, $\xi\in\{0.1,0.15,0.2\}$, $\upsilon_0\in\{0.8,0.9\}$, $|\mcd|\in \{10000,15000\}$, the number of hidden layers in the MLP of actor and critic networks $\in\{1,2\}$, and the hidden size in the MLP of actor and critic networks $\in\{64,128\}$. The hyperparameter configuration yielding the highest average actual return during the evaluation phase is reported. This average actual return is determined as follows:
 \begin{equation}\label{eq:actual_return}
 	AR = \frac{1}{M} \sum_{j=1}^{M} \sum_{i=1}^{V(j)}{\gamma^{i-1} r_{k^\star \oplus i,t^\star+\Delta t^i}^{j}}\;,
 \end{equation}
where $r_{k^\star \oplus i,t^\star+\Delta t^i}^{j}$ is the reward received when bus $b_j$ departs from stop $k^\star \oplus i$ at time $t^\star + \Delta t^i$ in the \textit{simulation}. Here, $k^\star$ represents the first stop from which bus $b_j$ departs during the simulation period, and $t^\star$ corresponds to the departure time. With a slight abuse of notation, we use $\Delta t^i$ to denote the random time interval between departures from stop $k^\star \oplus i$ and stop $k^\star \oplus i \oplus 1$. $V(j)$ is the number of experienced transitions for bus agent $b_j$ during the simulation.\par

To stay conservative, the Offline-GMAC training process adopts the same hyperparameter values optimized for the Online-GMAC. ($\xi$, $\nu_0$ and $|\mathcal{D}|$ used in the Online-GMAC are not needed for the Offline-GMAC.) Only the exclusive hyperparameter $\eta$ is optimized via a grid search over the set $\{0.1, 0.5, 1, 1.5\}$ to fine-tune the Offline-GMAC when training on uncontrolled data. Eq.~\eqref{eq:actual_return} is again used as the objective function for determining the most suitable value of $\eta$.\par

All algorithms were implemented on a 2022 Macbook Air equipped with an M2 processor and 16GB memory. Training a single Online-GMAC policy under one hyperparameter configuration takes approximately half an hour. Note that most of the time is spent collecting data through simulator interaction. In contrast, the proposed Offline-GMAC approach can train a policy for one hyperparameter configuration within 1 minute.

\subsubsection{Evolution of headway variations and holding times in the initial stage}\label{sec:evolution_headway}
Figs.~\ref{fig:headway_evolution} and \ref{fig:hold_time_evolution} plot the evolution curves of headway variation and holding time, respectively, for the first 100 episodes of implementation under the Online-GMAC and the Offline-GMAC trained by uncontrolled data. These results were obtained using the hyperparameter configurations developed in Section \ref{sec:hyperparameter_selection}. Headway variation is defined as the standard deviation of departure headways, measured across all buses at all stops. Meanwhile, the holding time was calculated as the average across all stops and buses.\par

\begin{figure}[!ht]
    \centering
    \includegraphics[scale=0.23]{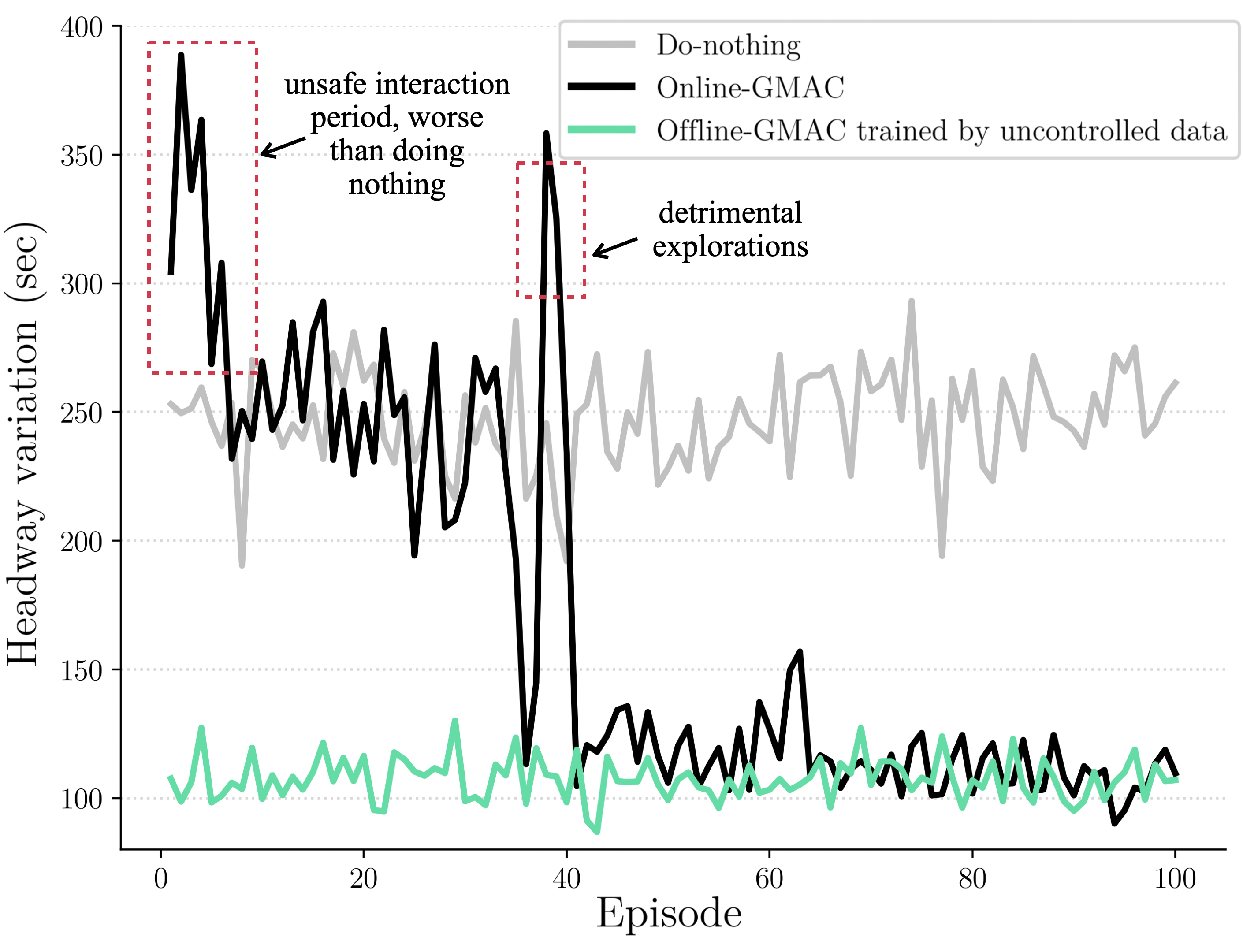}
    \caption{Headway variation evolution under different holding strategies.}\label{fig:headway_evolution}
\end{figure}\par
\begin{figure}[!ht]
    \centering
    \includegraphics[scale=0.23]{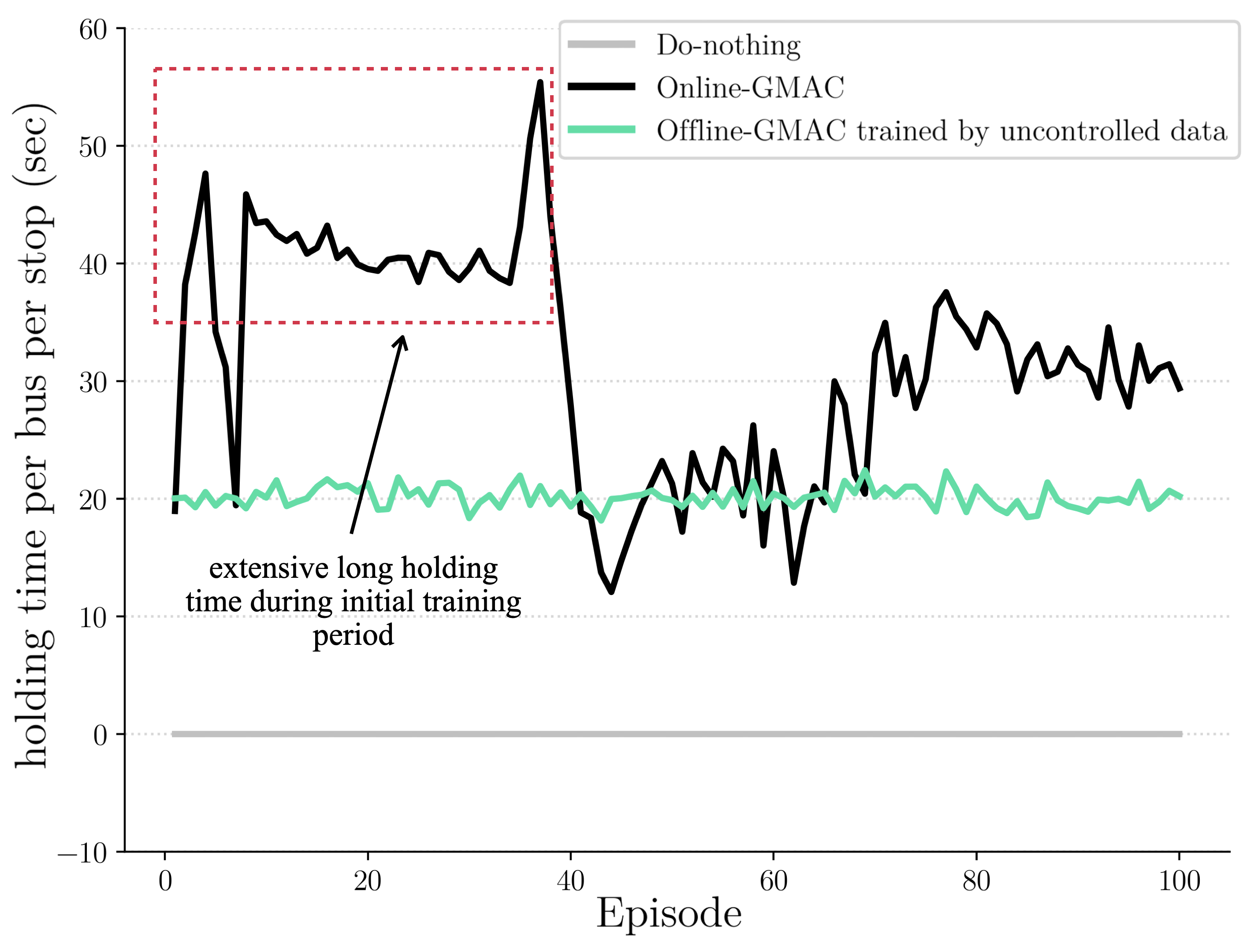}
    \caption{Holding time evolution under different holding strategies.}\label{fig:hold_time_evolution}
\end{figure}\par

Note that under each distinct strategy, the episodes serve different purposes.   Specifically, for the Offline-GMAC trained with uncontrolled data (depicted by the green curves), the 100 episodes were simulated under the trained policy. For the Online-GMAC (the black curves), the policy was continuously updated through training during the 100 episodes. In addition, results under the do-nothing strategy (the grey curves) are also presented, serving as a baseline reference. \par

First note in Figs.~\ref{fig:headway_evolution} and \ref{fig:hold_time_evolution} that the Online-GMAC performed poorly in the beginning because the agents had little prior knowledge of the environment. The headway variation even significantly exceeded that of the do-nothing strategy, and buses were held for extended times. This would lead to increased passenger dissatisfaction and elevated operating costs. The performance of the Online-GMAC gradually improved as the agents continued to interact with the environment, exploring more of the state-action space and understanding its associated rewards and penalties. Nevertheless, the costs of learning include approximately 40 days of compromised headway reliability and excessive holding times (recall that each episode represents a morning peak period). This period, marked by unreliable operations and frequent bus bunching, is unacceptable for real-world bus services.\par

In contrast, the Offline-GMAC, once trained by uncontrolled trajectory data, demonstrated high performance from the outset, exhibiting both reduced headway variations and shorter holding times. No further learning is required and the performance remained consistently at a high level throughout all episodes. This merit renders the Offline-GMAC particularly beneficial for real-world implementations. All we need is historical GPS trajectory data, a resource that is widely available.\par

\subsection{Performance comparison of various strategies over the medium to long run}\label{sec:pareto_comparison}
This section compares the Offline-GMAC policies trained by 50 episodes of uncontrolled and FHC-controlled data sets, respectively, the Online-GMAC policy trained for 100 episodes and the FHC strategy. Each strategy is executed in the simulator for an additional 100 episodes, during which their performances are measured and compared. No training is conducted in these episodes to ensure a fair comparison of the policies' performances over the medium to long run.\par

For a comprehensive and equitable comparison, we conduct grid searches on the essential hyperparameters of the above strategies and develop Pareto frontiers depicting headway variation versus holding time in Fig.~\ref{fig:pareto_headway}. Pareto frontier depicts the set of optimal trade-offs where one performance measure cannot be improved without worsening another. It allows us to identify the best strategy that balances the goals of minimizing headway variation while considering the associated holding time costs. The points used for constructing the Pareto frontier of each strategy are trained and evaluated as per the following hyperparameter settings:
\begin{enumerate}
	\item \textbf{Online-GMAC.} Vary $w$ from 0.001 to 0.02 with increments of 0.001 based on the optimal hyperparameter set established in Section \ref{sec:hyperparameter_selection};
	\item \textbf{Offline-GMAC trained by uncontrolled data.} Perform a grid search on $w$ from 0.001 to 0.01 with increments of 0.001 and $\eta\in \{0.1, 1, 1.5\}$, with reference to the optimal hyperparameter set in Section \ref{sec:hyperparameter_selection};
	\item \textbf{FHC strategy.} Search $\zeta$ from $0.1$ to $0.9$ with increments of $0.05$;
	\item \textbf{Offline-GMAC trained by FHC data.} Using training data generated by the FHC strategy with $\zeta=0.6$, perform a grid search on $w$ from 0.001 to 0.02, as per the optimal hyperparameter set in Section \ref{sec:hyperparameter_selection}.
\end{enumerate}
Headway variation and holding time are averaged across all buses, stops, and the 100 episodes. Passenger waiting time is averaged across all passengers and the 100 episodes. (For a detailed visualization of bus trajectories comparing uncontrolled and Offline-GMAC-controlled scenarios, please refer to Appendix~\ref{apdx:ts_visualization}.)\par

\begin{figure}[!ht]
    \centering
    \includegraphics[width=0.85\textwidth]{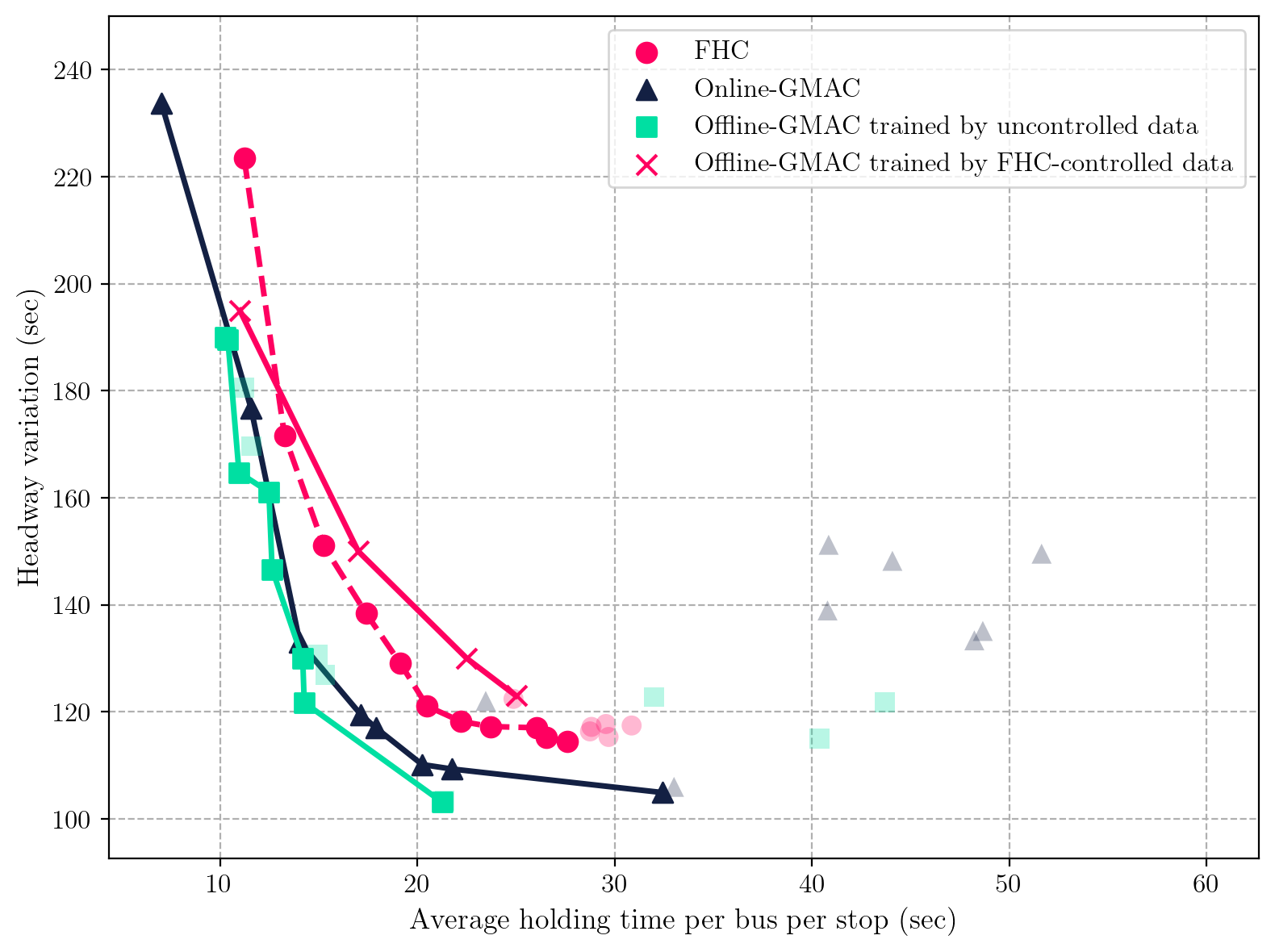}
    \caption{Pareto frontiers of headway variation versus holding time under various control strategies (after training).}\label{fig:pareto_headway}
\end{figure}\par

First note in Fig.~\ref{fig:pareto_headway} that the Pareto frontiers of the Online-GMAC and Offline-GMAC trained by uncontrolled data are positioned in the lower left corner compared to the FHC strategy. This indicates that the above two RL methods can attain the same headway regularity with shorter holding times. For example, to achieve a headway variation of 120 seconds, the Offline-GMAC trained using uncontrolled data can reduce the holding time by approximately 25\% versus the FHC strategy. Moreover, these two RL methods can produce policies with lower headway variations given sufficient holding times. For example, the Online- and Offline-GMAC can diminish headway variation to 103 and 104 seconds, respectively, while the minimum headway variation attained by the FHC strategy is 114 seconds. These results are aligned with the literature demonstrating the superior performance of online RL approaches over non-RL methods \citep{wang2020dynamic}.\par

Interestingly, Fig.~\ref{fig:pareto_headway} also reveals that the Offline-GMAC trained using uncontrolled data achieves comparable performance to the Online-GMAC. When the allowable average holding time is less than 22 seconds, the Offline-GMAC even slightly outperforms the Online-GMAC. On the other hand, when holding times are sufficiently long, the Online-GMAC can discover policies with the lowest headway variation, but the improvement tends to be small and is achieved at the cost of considerably longer holding times. The remarkable success of the Offline-GMAC trained by uncontrolled data may stem from its broader and more exploratory range of synthesized actions during training.\par

On the other hand, it is surprising that the Offline-GMAC trained with FHC data not only performs worse than the one trained with uncontrolled data, but is also inferior to the original FHC strategy. One plausible explanation for this outcome is that the state-action space induced by the FHC strategy might be smaller than the space created by our synthesis scheme. Note that the synthesized actions are derived from stochastic travel times and are independent of the states, rendering them more diverse than those collected under the FHC strategy where the actions are generated by a predefined rule. This finding indicates that employing rule-based expert policies for generating RL training data may not necessarily confer additional advantages. Contrarily, this hybrid method could be counterproductive since rule-based policies and the flexible, search-oriented reinforcement learning methods are two fundamentally different approaches.

As a robustness check, we additionally evaluated the trained Offline-GMAC policies in a modified simulator that accepts boardings while a bus is held. Using the same configurations as in Fig.~\ref{fig:pareto_headway}, the resulting Pareto frontier remains close to that of the original simulator. We also re-ran the same Pareto sweep under a perturbed route configuration in which each stop's passenger arrival rate is shifted by Gaussian noise (standard deviation 0.2 pax/min). Under this perturbed demand, the Online-GMAC and Offline-GMAC Pareto frontiers continue to overlap closely, indicating that the comparative findings in Fig.~\ref{fig:pareto_headway} are not specific to the original route configuration.

\subsection{Empirical justifications for the data synthesis scheme}\label{sec:closeness_verification}

To complement the theoretical analysis in Section~\ref{sec:theoretical_justification_data_synthesis}, we numerically verify the closeness between the observed distribution $p_{obs}$ and desired distribution $p_{des}$. We compare these distributions using data from our simulator for bus route No.~34.\par

The red and black curves in Fig.~\ref{fig:visual_p}a depict the probability density functions of the \textit{observed} forward spacing distribution associated with uncontrolled data, $p_{obs}(f_{k\oplus 1,t+\Delta t^j_k}^j \mid f_{k,t}^j=0.05, h_{k,t}^j=0)$, and the \textit{desired} distribution, $p_{des}( f_{k\oplus 1,t+\Delta t^j_k}^j \mid f_{k,t}^j=0.05, \tilde{h}_{k,t}^j )$, respectively. The latter curve is constructed by synthesizing $\tilde{h}_{k,t}^j$ from uncontrolled data and subsequently implementing $\tilde{h}_{k,t}^j$ in the simulator. Notably, these two distributions exhibit a substantial degree of similarity. The similar pattern is found for the backward spacing, as shown in Fig.~\ref{fig:visual_p}b. In particular, the two distributions align most closely on the lower side of the forward spacing and the upper side of the backward spacing, precisely the regions that the $\max\{0,\cdot\}$ truncation in Algorithm~\ref{algr:infer_a} does not act on. This finding partially explains the good performance of our Offline-GMAC policy trained by uncontrolled data.
\begin{figure}
  \centering
  \subfloat[]{%
       \includegraphics[width=0.45\textwidth]{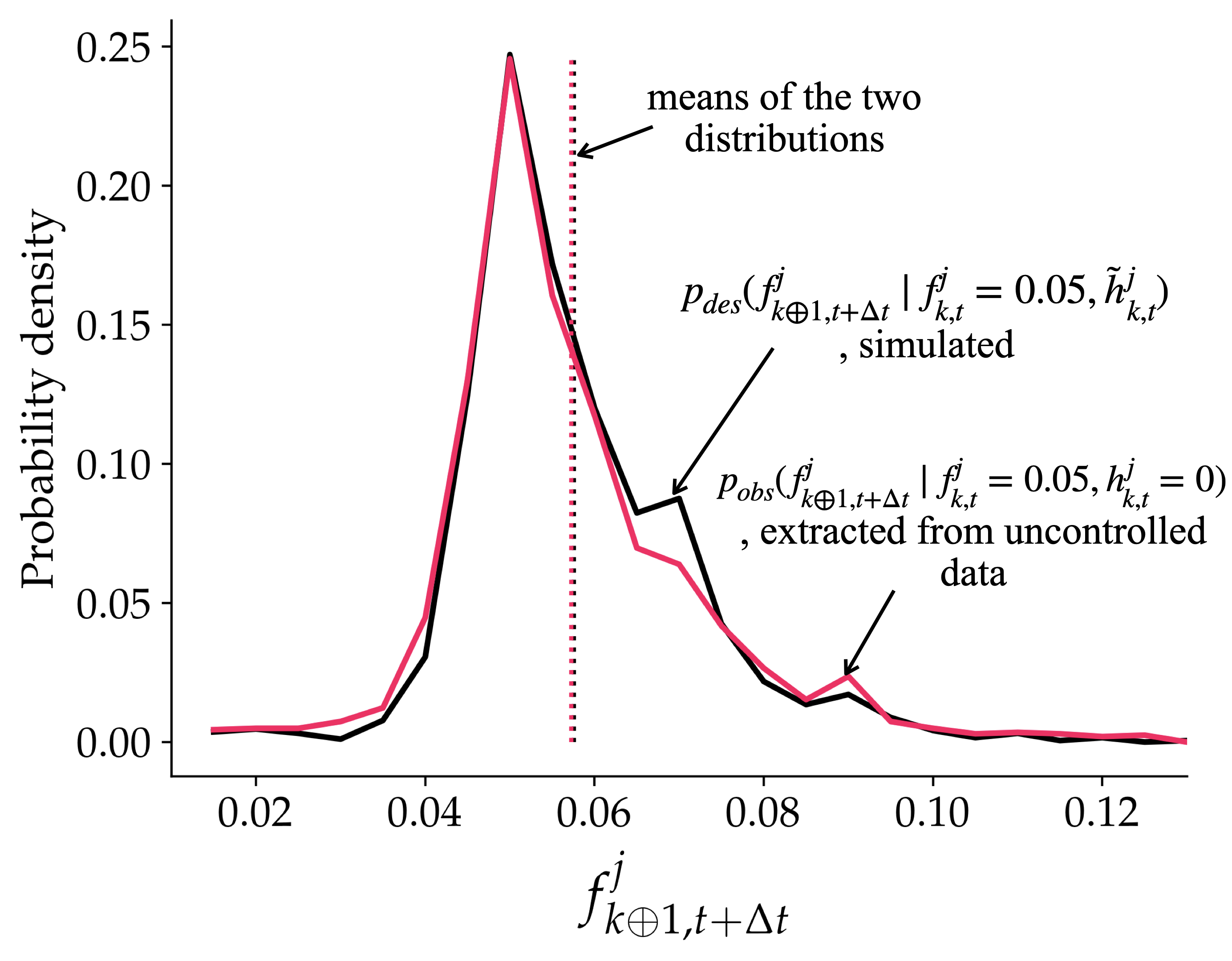}}
  \subfloat[]{%
        \includegraphics[width=0.45\textwidth]{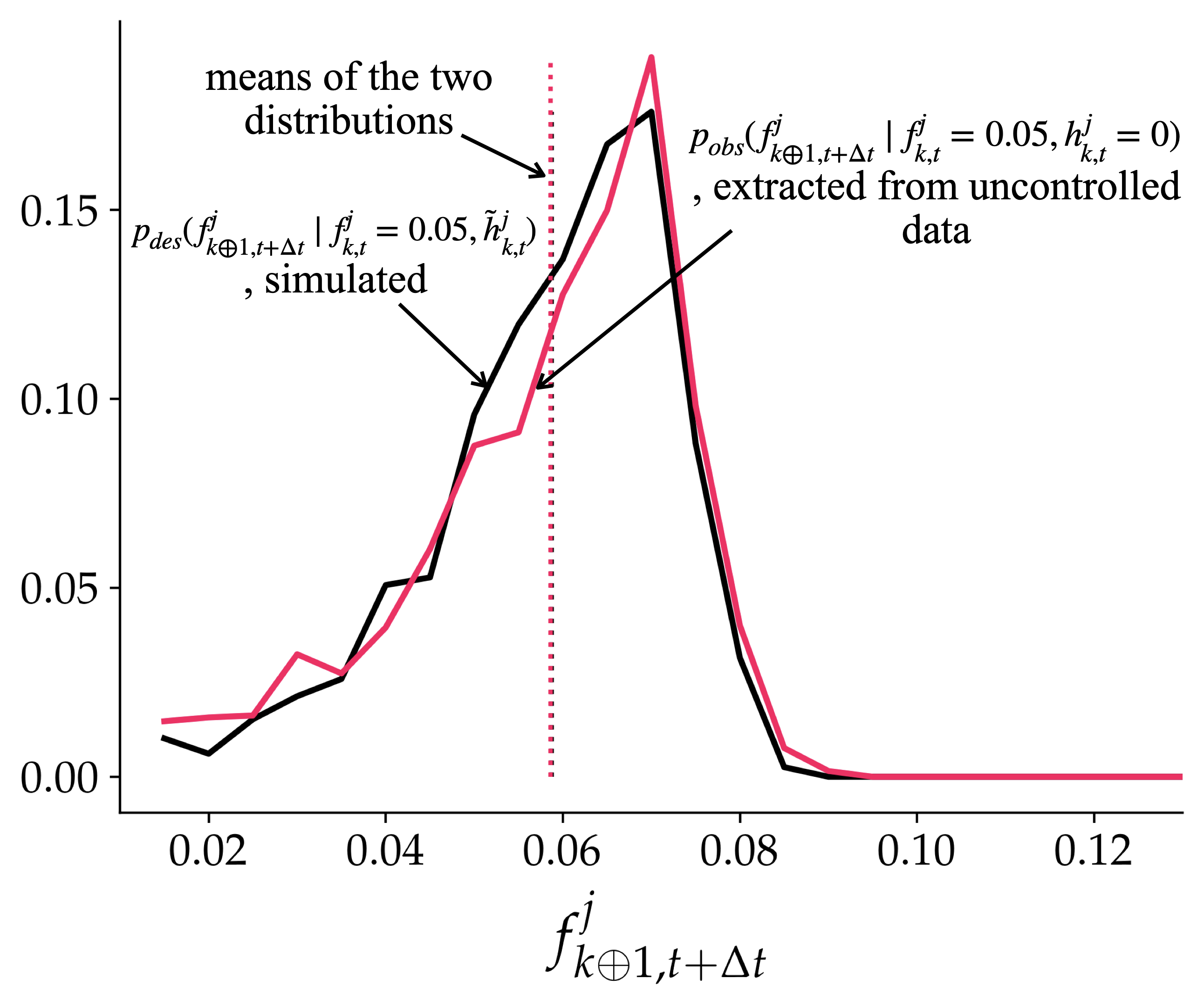}}
  \caption{Transition probability distributions of $f^j_{k\oplus 1,t+\Delta t}$ given $f_{k,t}^j=0.05$. (a) Forward spacing $f_{k,t}^j$. (b) Backward spacing $b_{k,t}^j$.}
  \label{fig:visual_p}
\end{figure}

\section{Issues on real-world deployment}\label{sec:deployment}

\subsection{Data size requirement}\label{sec:historical_data_size}
Bus route plans, stop locations, timetables, and demand patterns frequently change over time, necessitating contingent adjustments to control policies. It is therefore worthwhile to investigate the historical data size required to train an optimal policy before such an operational change occurs. To this end, we trained Offline-GMAC policies using an increasing number of episodes of uncontrolled data (from 1 to 25 episodes). For each dataset, training was performed with the optimal hyperparameters defined in Section \ref{sec:hyperparameter_selection}, with $w$ varying from 0.001 to 0.005 with step size of 0.001. The policy yielding the lowest headway variation was selected for each dataset. Fig.~\ref{fig:sa_day} reports the performance of these policies, assessed over 100 episodes following the training process.

\begin{figure}[!ht]
    \centering
    \includegraphics[scale=0.72]{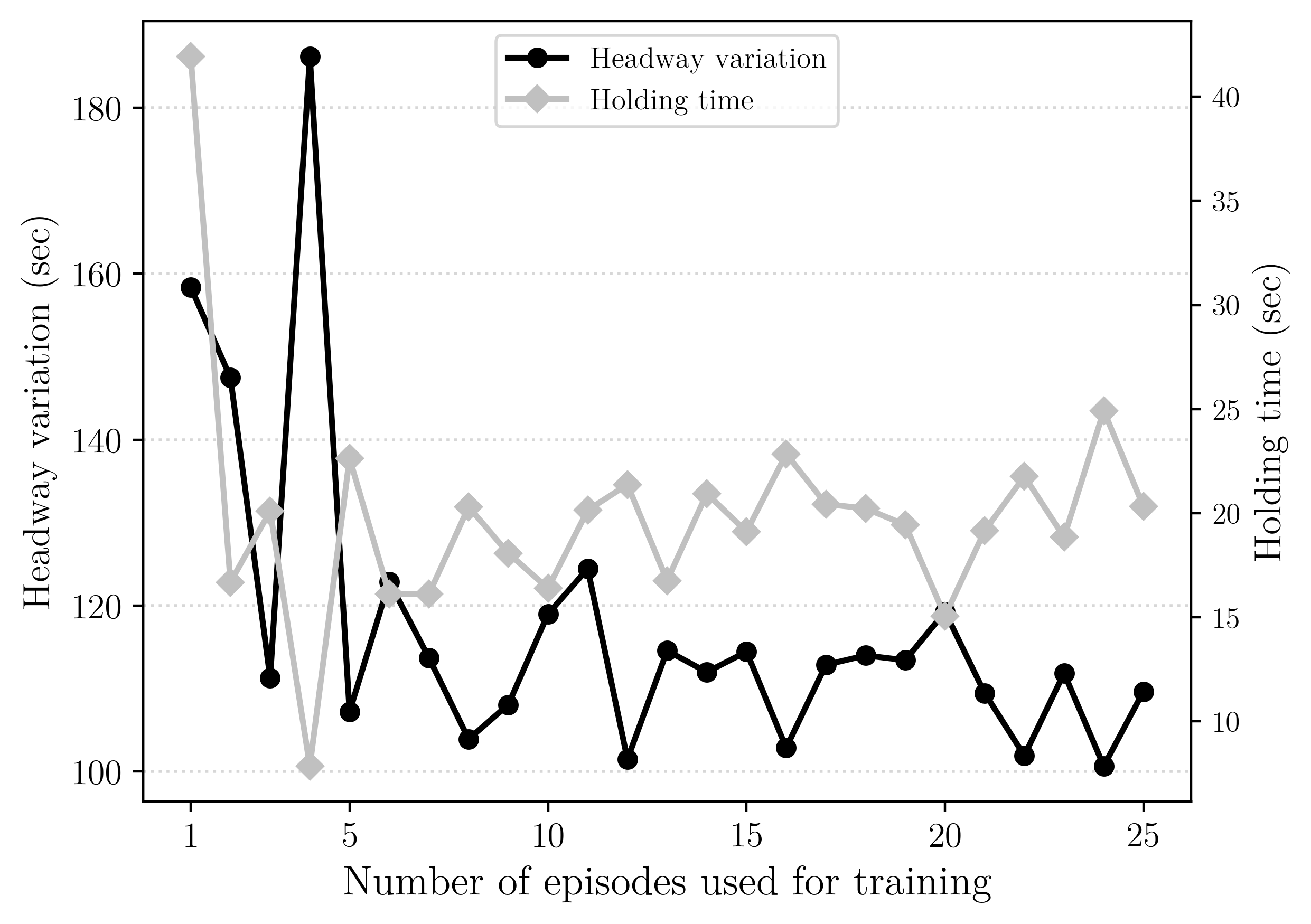}
    \caption{Impacts of training data size on the performance of Offline-GMAC.}
    \label{fig:sa_day}
\end{figure}\par 

Notably, the headway variation largely remains under 120 seconds once the training data size exceeds 5 episodes, while holding time fluctuates between 15-25 seconds. This means that even trained with a smaller dataset containing only a few episodes, Offline-GMAC can achieve a performance comparable to policies trained on much more data. This modest data size requirement is favorable for real-world deployment.

\subsection{Hyperparameter selection}\label{sec:model_selection}

In previous sections, optimal hyperparameters were determined by simulating the resulting policies' performances in the ``real environment''---a common approach in the offline RL literature, as noted by \citet{paine2020hyperparameter}. However, for bus control problems, it's unwelcome to directly evaluate policies in the real environment since underperforming policies could potentially disrupt the service. Therefore, a safer offline assessment approach for hyperparameter selection, reliant solely on historical data, is necessary to identify well-performing policies prior to their real-world implementation.\par

In this section, we conduct offline hyperparameter selection (OHS), meaning to evaluate and pick the best policy from a set of policies trained using different hyperparameters, without interacting with the real environment. (OHS is closely related to offline policy evaluation (OPE), which estimates the value function of a target policy using data collected by a different behavior policy.) We follow the workflow of OHS in \citet{paine2020hyperparameter}:
 \begin{enumerate}
 	\item Train offline RL policies using diverse hyperparameter settings;
 	\item Compute scalar statistics summarizing the performance of each policy, without interacting with the environment;
 	\item Pick the $m$ best policies according to the summary scalar statistics and execute them in the real environment. (For practical implementation, we can execute the $m$ policies daily, select the best one based on performance, and continue to run the best policy over the long term.)
 \end{enumerate}\par
 
In this paper, we leverage the average critic value from the final offline training epoch (refer to Algorithm \ref{algr:training}) as the scalar statistic for policy selection. To assess the effectiveness of the aforementioned workflow, we train a range of policies offline using the same uncontrolled data set in Sections \ref{sec:comparison_strategies}-\ref{sec:pareto_comparison}, each under a distinct hyperparameter configuration. Each policy is then tested through 100 episodes of simulation. (Note again that the usage of the simulator here is only for the assessment purpose.) The average actual return calculated by Eq.~\eqref{eq:actual_return} of each episode is averaged across all 100 episodes.\par

Fig.~\ref{fig:rank_correlation} plots the average critic value of the last offline training epoch against the average actual return. Each data point represents a policy trained using a distinct set of hyperparameters. Note that all the critics tend to underestimate the return. This trend likely arises because the uncontrolled data predominantly cover low-reward regions within the state-action space, causing the critics to be inherently pessimistic. This pessimism is opposite to the over-optimism more commonly reported in offline value-based methods \citep{kumar2020conservative}, where the critic produces inflated estimates on out-of-distribution actions rather than fitting low-reward in-distribution data as in our setting. However, the actual return and the estimated critic value exhibit a positive correlation, with a Spearman's rank correlation coefficient of 0.67 and an associated $p$-value of $2\mathrm{e}-8$.\par

\begin{figure}[!ht]
    \centering
    \includegraphics[scale=0.75]{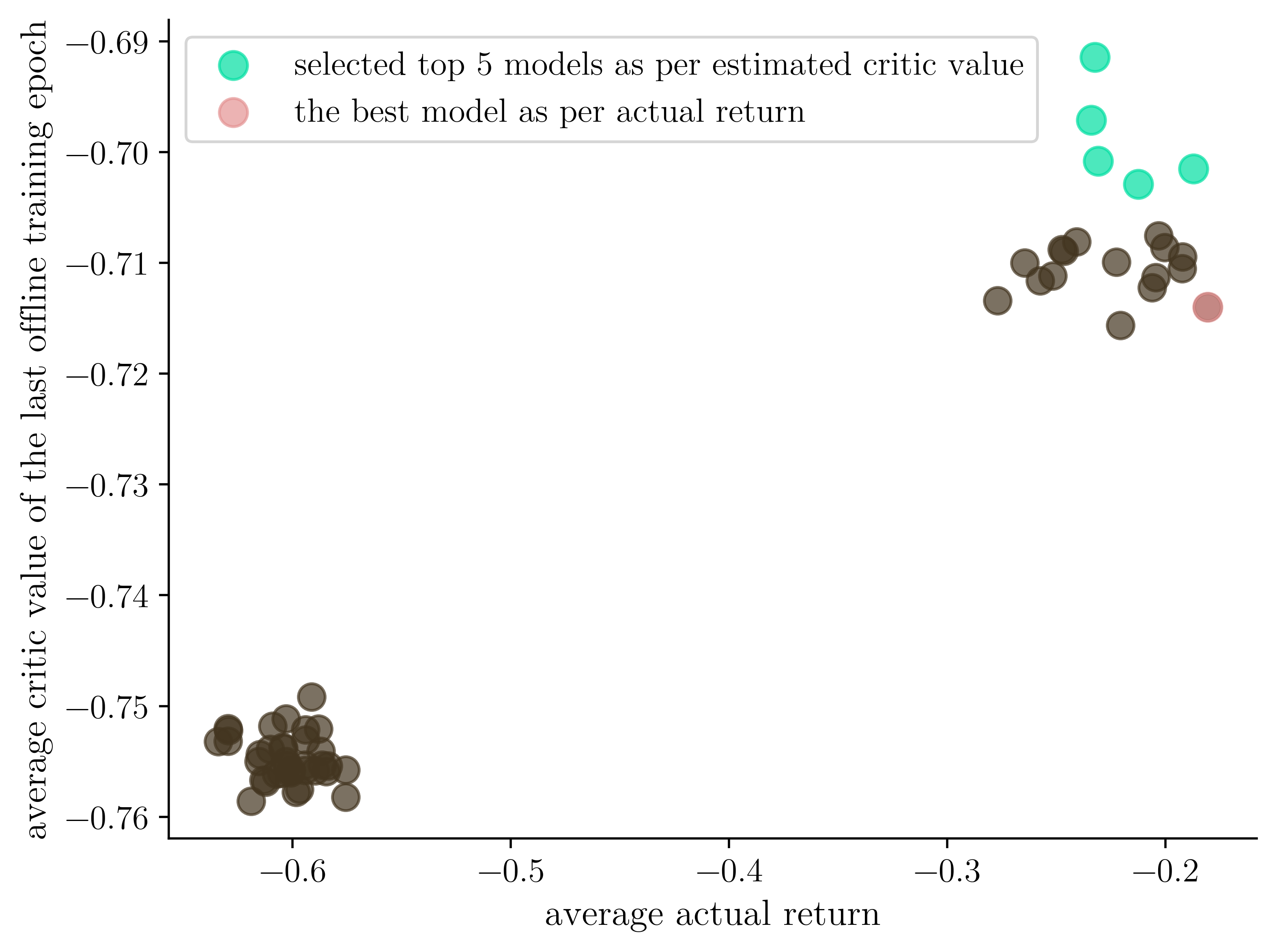}
    \caption{Critic value of the last training epoch versus average actual return.}
    \label{fig:rank_correlation}
\end{figure}\par

Furthermore, the data points are separated into two clear clusters. The points in the lower left corner exhibit very low average actual returns, meaning that agents associated with these hyperparameter configurations fail to learn useful control knowledge from the offline data during training. In contrast, most points clustered in the upper right achieve returns concentrated between -0.25 and -0.18. The clear bifurcation into high versus low return policies demonstrates that, once offline training is complete, the critic value alone can effectively distinguish between good and bad policies.\par

Finally, selecting the top five policies (the green points in Fig.~\ref{fig:rank_correlation}) as per the estimated critic values is sufficient to provide a well-performing policy. While the optimal policy (the red point) generates a headway variation of 96 seconds with an average holding time of 28 seconds, the top five selected policies (ranked by estimated critic value) yield headway variations of 125, 111, 117, 98, and 113 seconds, with average holding times of 15, 18, 17, 28, and 44 seconds, respectively. If the Pareto frontiers under these five policies were plotted in Fig.~\ref{fig:pareto_headway} and compared to Online-GMAC, we would observe that the first four demonstrate better headway regularity than the Online-GMAC for given holding time thresholds.

\section{Conclusions and future research directions}\label{sec:conclusion}
In this paper, we propose a data-driven formulation of multi-agent RL, Offline-GMAC, to learn an optimal bus holding policy from logged uncontrolled bus trajectory data, without interaction with the environment. The method adopts the centralized training and decentralized execution paradigm: the centralized critic leverages GNNs to perform effective credit assignments by observing the states and actions of all agents; and the decentralized actor makes holding decisions based on local observations. The method introduces a novel data augmentation scheme that exploits a key equivalence: stochastic delays in link travel time produce effects identical to those of deliberate holding interventions. The scheme synthesizes holding actions from historical GPS trajectories, transforming uncontrolled data into usable transition tuples for training holding policies.\par

Under the Lipschitz continuity assumption on the underlying system dynamics, we prove that the 1-Wasserstein distance between the distribution of synthesized state transitions and the distribution of true controlled state transitions admits an upper bound. This bound scales with the standard deviation of link travel time plus a tunable term related to the synthesis process itself, all multiplied by the system's Lipschitz constant. The bound's explicit form reveals how synthesis accuracy degrades with increased travel time uncertainty, while also identifying the regulator parameter $\eta$ as a mechanism for mitigating distributional shift.\par

Evaluations using a simulator driven by real-world bus operation data demonstrate that the Offline-GMAC method can learn a policy that matches the performance of its online counterpart. Furthermore, the distribution shift problem, commonly encountered in offline RL, appears to be less problematic when using the synthesized training data.\par

Through extensive numerical experiments, we find our proposed Offline-GMAC method has the following deployment advantages:
\begin{enumerate}[(i)]
	\item By leveraging offline training, we can develop and refine holding policies without the need for online interactions. This reduces the potential risks and disturbances caused by deploying untested policies directly in the real world.
	\item Our method eliminates the need for a simulator used for offline training, mitigating the challenges related to parameter estimation and complex system modeling that often arise when creating a high-fidelity simulator.
	\item Additionally, the data requirement is minimal, as our offline method can generate a dependable holding policy using only a few days' worth of historical data. This enables rapid adaptations to frequent demand and operational changes in the bus route.
	\item Lastly, the offline nature of our method enables fast training of customized holding policies using historical data of different time periods. Distinct policies tailored to each period can be learned within minutes. This computational efficiency further enhances the practicality of our offline method. In contrast, online methods require repeatedly interacting with the environment to collect new data for training the agents, resulting in a time-consuming training process.
\end{enumerate}

There are three potential future research directions. First, a real-world bus corridor is often shared by multiple bus lines, enabling common-line passengers to choose buses among multiple lines \citep{cominetti2001common,schmocker2016bus,laskaris2018multiline}. Future efforts can go to developing more sophisticated multi-agent RL algorithms to consider the interactions among lines. Second, our proposed method can be extended to incorporate other control measures, including stop-skipping, bus speed control \citep{he2019approach}, and signal priority schemes \citep{estrada2016bus,anderson2020effect}. Unlike holding, our travel-time-based synthesis scheme does not apply to these measures, and action-labeled data may need to come from expert demonstrations (e.g., rule-based controllers) or other sources. Finally, we are seeking opportunities for real-world tests through collaboration with the local transit agency.

\section*{Acknowledgements}
This work was supported by the National Natural Science Foundation of China (Project Nos.\ 72201214, 72571220), and the General Research Funds (Project Nos.\ 15217415, 15237624) provided by the Research Grants Council of Hong Kong.

\newpage
\begin{appendices}
\renewcommand{\theHequation}{\Alph{section}.\arabic{equation}}
\renewcommand{\theHfigure}{\Alph{section}.\arabic{figure}}
\renewcommand{\theHlemma}{\Alph{section}.\arabic{lemma}}

\section{Proof of Lemma~\ref{lemma:distance_of_arrivals}}\label{apdx:lemma_distance_of_arrivals_proof}

\setcounter{equation}{0}\renewcommand{\theequation}{A.\arabic{equation}}

Let $\hat{\epsilon}_{k,t}^j \equiv T_{k,t}^j - \mu_k$ and $\bar{\epsilon}_{k,t}^j \equiv Y_{k,t}^j - \mu_k$ be the random variables for the actual and historical travel time deviations, respectively.\par

The 1-Wasserstein distance, $\mathcal{W}_1$, is bounded by the expected distance under the independent coupling of the two random variables:
\begin{equation}
	\mathcal{W}_1(P(\hat{t}_{\text{arr},k\oplus 1}^{j}), P(\bar{t}_{\text{arr},k\oplus 1}^{j})) \le \mathbb{E}[|\hat{t}_{\text{arr},k\oplus 1}^{j} - \bar{t}_{\text{arr},k\oplus 1}^{j}|].
\end{equation}
Let $C_h = \frac{W}{\eta\cdot\sigma_k}$. Substituting the definitions for the arrival times into the right-hand side of this inequality, which we denote by $E_h$, gives:
\begin{equation}
	E_{h} = \mathbb{E}[ | (t_{\text{dep},k}^{j} + h + T_{k,t}^j) - (t_{\text{dep},k}^{j} + Y_{k,t}^j) | ] = \mathbb{E}[ | h + \hat{\epsilon}_{k,t}^j - \bar{\epsilon}_{k,t}^j | ],
\end{equation}
where the expectation is a joint expectation over $\hat{\epsilon}_{k,t}^j$ and $\bar{\epsilon}_{k,t}^j$, conditioned on the event $H_s = h$. From the definition of $H_s$, the condition $H_s = h$ gives us $h = C_h \cdot \max\{0,\bar{\epsilon}_{k,t}^j\}$. Substituting this into the expectation gives:
\begin{equation}
	E_{h} = \mathbb{E}[ | C_h \cdot \max\{0, \bar{\epsilon}_{k,t}^j\} + \hat{\epsilon}_{k,t}^j - \bar{\epsilon}_{k,t}^j | ]. 
\end{equation}
We analyze $E_h$ using the law of total expectation, conditioning on the sign of $\bar{\epsilon}_{k,t}^j$.

\textbf{Case 1: The historical bus travel time is longer than the mean ($\bar{\epsilon}_{k,t}^j > 0$ and $h > 0$).}
The conditional expectation is $\mathbb{E}[|(C_h - 1)\bar{\epsilon}_{k,t}^j + \hat{\epsilon}_{k,t}^j| \mid \bar{\epsilon}_{k,t}^j > 0]$. By the triangle inequality, it is bounded by:
\begin{equation}
	\mathbb{E}[|(C_h - 1)\bar{\epsilon}_{k,t}^j + \hat{\epsilon}_{k,t}^j| \mid \bar{\epsilon}_{k,t}^j > 0] \le \mathbb{E}[|(C_h - 1)\bar{\epsilon}_{k,t}^j| \mid \bar{\epsilon}_{k,t}^j > 0] + \mathbb{E}[|\hat{\epsilon}_{k,t}^j| \mid \bar{\epsilon}_{k,t}^j > 0].
\end{equation}

As proven in Lemma~\ref{lemma:proportionality} (Appendix~\ref{apdx:lemma}), $\mathbb{E}[|\bar{\epsilon}_{k,t}^j| \mid \bar{\epsilon}_{k,t}^j > 0]$ is proportional to $\sigma_k$. Moreover, because $\hat{\epsilon}_{k,t}^j$ is independent of $\bar{\epsilon}_{k,t}^j$, $\mathbb{E}[|\hat{\epsilon}_{k,t}^j| \mid \bar{\epsilon}_{k,t}^j > 0]$ simplifies to the mean absolute deviation (MAD) $\mathbb{E}[|\hat{\epsilon}_{k,t}^j|]$, which is also proportional to $\sigma_k$. Let $\mathbb{E}[|\bar{\epsilon}_{k,t}^j| \mid \bar{\epsilon}_{k,t}^j > 0] = C_1 \sigma_k$ and $\mathbb{E}[|\hat{\epsilon}_{k,t}^j|] = C_{\text{mad}} \sigma_k$, where $C_1$ and $C_{\text{mad}}$ are positive constants derived from the standardized distribution $Z = \epsilon / \sigma_k$. The bound for Case 1 is: $C_1 |C_h - 1| \sigma_k + C_{\text{mad}} \sigma_k$.\par

\textbf{Case 2: The historical bus travel time is shorter than or equal to the mean ($\bar{\epsilon}_{k,t}^j \le 0$ and $h = 0$).}
The conditional expectation is $\mathbb{E}[|\hat{\epsilon}_{k,t}^j - \bar{\epsilon}_{k,t}^j| \mid \bar{\epsilon}_{k,t}^j \le 0]$. Using the triangle inequality, this is bounded by:
\begin{align}
	\mathbb{E}[|\hat{\epsilon}_{k,t}^j - \bar{\epsilon}_{k,t}^j| \mid \bar{\epsilon}_{k,t}^j \le 0] \leq \mathbb{E}[|\hat{\epsilon}_{k,t}^j| + |\bar{\epsilon}_{k,t}^j| \mid \bar{\epsilon}_{k,t}^j \le 0] \nonumber \\
	= \mathbb{E}[|\hat{\epsilon}_{k,t}^j| \mid \bar{\epsilon}_{k,t}^j \le 0] + \mathbb{E}[|\bar{\epsilon}_{k,t}^j| \mid \bar{\epsilon}_{k,t}^j \le 0].
\end{align}\par

Again, due to independence, the first term is $\mathbb{E}[|\hat{\epsilon}_{k,t}^j|] = C_{\text{mad}} \sigma_k$. The second term, $\mathbb{E}[|\bar{\epsilon}_{k,t}^j| \mid \bar{\epsilon}_{k,t}^j \le 0]$, is also proportional to $\sigma_k$, as per Lemma~\ref{lemma:proportionality} (Appendix~\ref{apdx:lemma}). Let this be $C_2 \sigma_k$. The bound for Case 2 is: $C_{\text{mad}} \sigma_k + C_2 \sigma_k = (C_{\text{mad}} + C_2) \sigma_k$.\par

The total expectation $E_h$ is a probability-weighted sum of the bounds from the two cases. Let $p_1 = \mathbb{P}(\bar{\epsilon}_{k,t}^j > 0)$ and $p_2 = \mathbb{P}(\bar{\epsilon}_{k,t}^j \le 0)$. These probabilities are constants that do not depend on $\sigma_k$.
\begin{align}
	E_h & \le p_1 \cdot ( C_1 |W/\eta - \sigma_k| + C_{\text{mad}} \sigma_k ) + p_2 \cdot ( (C_{\text{mad}} + C_2) \sigma_k ) \\
	&= (p_1 C_1) |W/\eta - \sigma_k| + (C_{\text{mad}} + p_2 C_2) \sigma_k.
\end{align}\par

Since $p_1, p_2, C_1, C_2, C_{\text{mad}}$ are all positive constants, the entire expression is bounded by a term of the order $\mathcal{O}(|W/\eta - \sigma_k| + \sigma_k)$.

\section{Proof of Theorem~\ref{theorem:push_forward}}\label{apdx:theorem_push_forward}

\setcounter{equation}{0}\renewcommand{\theequation}{B.\arabic{equation}}

By the dual formulation of $\mathcal{W}_1$ in Eq.~\eqref{eq:dual_formulation}, for all $P_1, P_2 \in \mathcal{P}_1(\mathcal{Y})$:
		\begin{equation}\label{eq:to_analyze}
			\mathcal{W}_1(K_\#P_1, K_\#P_2) = \sup_{\|\psi\|_{\text{Lip}} \le 1} \left\{  \int_{\mathcal{X}} \psi(x) \, d(K_\#P_1)(x) - \int_{\mathcal{X}} \psi(x) \, d(K_\#P_2)(x)  \right\}.
		\end{equation}
		According to the definition of the the mixture distribution $K_\#P$ in Eq.~\eqref{eq:mix_definition} and applying Fubini's theorem, we have:
		\begin{equation}
			\int_{\mathcal{X}} \psi(x) \, d(K_\#P)(x) = \int_{\mathcal{Y}} \left( \int_{\mathcal{X}} \psi(x) \, dK(y)(x) \right) \, dP(y), \|\psi\|_{\text{Lip}} \le 1.
		\end{equation}
		Define a new function $h_{\psi}: \mathcal{Y} \to \mathbb{R}$ as the inner integral, which depends on our chosen $\psi$:
		\begin{equation}
			h_{\psi}(y) \equiv \int_{\mathcal{X}} \psi(x) \, dK(y)(x), \|\psi\|_{\text{Lip}} \le 1.
		\end{equation}
		
		Substituting this $h_{\psi}(y)$ back into Eq.~\eqref{eq:to_analyze}, we have:
		\begin{equation}
			\mathcal{W}_1(K_\#P_1, K_\#P_2) = \sup_{\|\psi\|_{\text{Lip}} \le 1} \left\{ \int_{\mathcal{Y}} h_{\psi}(y) \, dP_1(y) - \int_{\mathcal{Y}} h_{\psi}(y) \, dP_2(y) \right\}.
		\end{equation}
		
		We next determine the Lipschitz constant of $h_{\psi}(y)$ where $\|\psi\|_{\text{Lip}} \le 1$. For any $y_1$ and $y_2$,
		\begin{align}
			&|h_{\psi}(y_1) - h_{\psi}(y_2)| = \left| \int_{\mathcal{X}} \psi(x) \, dK(y_1)(x) - \int_{\mathcal{X}} \psi(x) \, dK(y_2)(x) \right| \nonumber \\
			&\leq \sup_{\|\psi\|_{\text{Lip}} \le 1} \left\{ \left| \int_{\mathcal{X}} \psi(x) \, d\mu(x) - \int_{\mathcal{X}} \psi(x) \, d\nu(x) \right| \right\} = \mathcal{W}_1(K(y_1), K(y_2)).
		\end{align}
		
		Invoking the theorem's premise \eqref{eq:kernel_assumption}, we get:
		\begin{equation}
			|h_{\psi}(y_1) - h_{\psi}(y_2)| \le L_s \cdot d_{\mathcal{Y}}(y_1, y_2),
		\end{equation}
		which shows that for any 1-Lipschitz $\psi$, the resulting function $h_{\psi}(y)$ is $L_s$-Lipschitz.
		
		Since $h_{\psi}(y)$ is $L_s$-Lipschitz, a ``scaled'' function $\tilde{h}_{\psi}(y) \equiv h_{\psi}(y) / L_s$ is 1-Lipschitz. (We assume $L_s > 0$; if $L_s=0$, the proof is trivial.) We substitute $h_{\psi}(y) = L_s \cdot \tilde{h}_{\psi}(y)$ and pull the constant $L_s$ out of the supremum:
		\begin{align}\label{eq:scaled_supremum}
			\mathcal{W}_1(K_\#P_1, K_\#P_2) &= \sup_{\|\psi\|_{\text{Lip}} \le 1} \left\{ \int (L_s \cdot \tilde{h}_{\psi}(y)) \, dP_1(y) - \int (L_s \cdot \tilde{h}_{\psi}(y)) \, dP_2(y) \right\} \nonumber \\
			&= L_s \cdot \sup_{\|\psi\|_{\text{Lip}} \le 1} \left\{ \int \tilde{h}_{\psi}(y) \, dP_1(y) - \int \tilde{h}_{\psi}(y) \, dP_2(y) \right\}.
		\end{align}
		Let $\tilde{H} = \{ \tilde{h}_{\psi} \mid \|\psi\|_{\text{Lip}} \le 1 \}$ be the set of all our scaled functions. $\tilde{H}$ is a set of 1-Lipschitz functions, which is a subset of the set of all 1-Lipschitz functions on $\mathcal{Y}$ (let us call this full set $\tilde{\Phi}$). The supremum over a subset ($\tilde{H}$) can be no larger than the supremum over the full, containing set ($\tilde{\Phi}$). Therefore:
		\begin{equation}\label{eq:two_supremum}
			\sup_{\tilde{h}_{\psi} \in \tilde{H}} \left\{ \int \tilde{h}_{\psi}(y) \, dP_1(y) - \int \tilde{h}_{\psi}(y) \, dP_2(y) \right\} \le \sup_{\tilde{\phi} \in \tilde{\Phi}} \left\{ \int \tilde{\phi}(y) \, dP_1(y) - \int \tilde{\phi}(y) \, dP_2(y) \right\}.
		\end{equation}
		The expression on the right is, by the dual definition \eqref{eq:dual_formulation}, the 1-Wasserstein distance $\mathcal{W}_1(P_1, P_2)$. By combining the equality from Eq.~\eqref{eq:scaled_supremum} with the inequality from Eq.~\eqref{eq:two_supremum}, we get our final result:
		\begin{equation}
			\mathcal{W}_1(K_\#P_1, K_\#P_2) \le L_s \cdot \mathcal{W}_1(P_1, P_2).
		\end{equation}

\section{A lemma on a scaling property of conditional expectation}\label{apdx:lemma}

\setcounter{equation}{0}\renewcommand{\theequation}{C.\arabic{equation}}
\setcounter{lemma}{0}\renewcommand{\thelemma}{C.\arabic{lemma}}

\begin{lemma}\label{lemma:proportionality}
Let $X$ be a continuous random variable with mean $\mathbb{E}[X] = 0$ and variance $\text{Var}(X) = \sigma^2$. We assume the standard deviation is finite and positive, i.e., $0 < \sigma < \infty$. Then, there exists a constant $C > 0$ that does not depend on $\sigma$, such that:
\begin{equation}
	\mathbb{E}[|X| \mid X > 0] = C \cdot \sigma.
\end{equation}

\end{lemma}

\begin{proof}
	Let $f_X(x)$ be the probability density function (PDF) of $X$. By the definition of conditional expectation for a continuous variable, we have:
	\begin{equation}\label{eq:definition_conditional_expectation}
		\mathbb{E}[|X| \mid X > 0] = \frac{\int_{0}^{\infty} |x| f_X(x) dx}{\mathbb{P}(X > 0)} = \frac{\int_{0}^{\infty} x f_X(x) dx}{\int_{0}^{\infty} f_X(x) dx}.
	\end{equation}
	We define a standardized random variable $Z = X / \sigma$. We find the relationship between $f_X(x)$, and the PDF of $Z$, $f_Z(z)$ as follows:
	\begin{equation}
		f_Z(z) = \frac{d}{dz} F_Z(z) = \frac{d}{dz} F_X(z\sigma) = f_X(z\sigma) \cdot \sigma.
	\end{equation}
	For any $x$, let $z = x/\sigma$. We thus have $f_X(x) = \frac{1}{\sigma} f_Z(x / \sigma)$. Substituting this expression for $f_X(x)$ into both the numerator and denominator of Eq.~\eqref{eq:definition_conditional_expectation}, we obtain:
	\begin{equation}\label{eq:condional_expectation_frac_origin_form}
		\mathbb{E}[|X| \mid X > 0] = \frac{\int_{0}^{\infty} x \left( \frac{1}{\sigma} f_Z(x / \sigma) \right) dx}{\int_{0}^{\infty} \left( \frac{1}{\sigma} f_Z(x / \sigma) \right) dx}.
	\end{equation}
	Since $z = x/\sigma$, we have $x = z\sigma$ and $dx = \sigma dz$. Substituting these relationships into Eq.~\eqref{eq:condional_expectation_frac_origin_form} yields:
	\begin{equation}\label{eq:condional_expectation_frac_final_form}
		\mathbb{E}[|X| \mid X > 0] = \frac{\int_{z=0}^{\infty} (z\sigma) \left( \frac{1}{\sigma} f_Z(z) \right) (\sigma dz)}{\int_{z=0}^{\infty} \left( \frac{1}{\sigma} f_Z(z) \right) (\sigma dz)} = \frac{\sigma \int_{0}^{\infty} z f_Z(z) dz}{\int_{0}^{\infty} f_Z(z) dz}.
	\end{equation}
	We define the constant $C\equiv \frac{\int_{0}^{\infty} z f_Z(z) dz}{\int_{0}^{\infty} f_Z(z) dz}$.
	
	The variable $Z = X / \sigma$ has variance $\text{Var}(Z) = \text{Var}(X / \sigma) = \frac{1}{\sigma^2}\text{Var}(X) = \frac{1}{\sigma^2}\sigma^2 = 1$. Here, $f_Z(z)$ is the PDF of a distribution with variance 1 and does not contain $\sigma$ as a parameter. Since $C$ is computed entirely from $f_Z(z)$, $C$ is a constant that does not depend on $\sigma$.
	
	The denominator of $C$ is $\int_{0}^{\infty} f_Z(z) dz = \mathbb{P}(Z > 0)$. Since $\mathbb{E}[Z] = 0$ and $\text{Var}(Z) = 1 > 0$, $Z$ is non-degenerate and must take on positive values to have a zero mean, so $0 < \mathbb{P}(Z > 0) \le 1$. The denominator is thus finite and strictly positive. The numerator of $C$, $\int_{0}^{\infty} z f_Z(z) dz$, is a component of the first absolute moment, $\mathbb{E}[|Z|]$. $\text{Var}(Z)=1$ implies that $\mathbb{E}[|Z|]$ is also finite, so the numerator is finite. Furthermore, since $\mathbb{P}(Z > 0) > 0$, $f_Z(z)$ must be positive on some interval in $(0, \infty)$, which ensures the numerator is also strictly positive. As $C$ is the ratio of two finite positive constants, $C$ is itself a finite and positive constant.
\end{proof}

\section{Estimating passenger arrival rates for Route No.~34}\label{apdx:arrival_rate_estimation}

\setcounter{equation}{0}\renewcommand{\theequation}{D.\arabic{equation}}

To calibrate passenger arrival rate for a stop $k$, $\lambda_k$, we utilize a historical data set of GPS trajectories spanning a week (March 11-17, 2021). The arrival and departure times for every bus $b_j$ at stop $c_k$, denoted by $t_{\text{arr},k}^j$ and $t_{\text{dep},k}^j$, respectively, are extracted. The $\lambda_k$ is estimated using the following equation:
\begin{equation}\label{eq:demand_estimation}
	t_{\text{dep},k}^j - t_{\text{arr},k}^j = \frac{\lambda_k}{\omega_k} (t_{\text{dep},k}^j - t^{j\ominus 1}_{\text{dep},k} ),
\end{equation}
where $t_{\text{dep},k}^j - t_{\text{arr},k}^j$ is bus $b_j$'s dwell time at stop $c_k$, and $t_{\text{dep},k}^j - t_{\text{dep},k}^{j\ominus 1}$ is the forward (departure) headway between buses $b_j$ and $b_{j\ominus 1}$. The underlying assumptions are: (1) passengers arriving at stop $c_k$ during the forward headway will board bus $b_j$ at a rate of $\omega_k$; and (2) bus $b_j$'s dwell time is dictated by boarding passengers since alighting moves are quick \citep{daganzo2009headway}. Note that the latter assumption is only employed in Eq.~\eqref{eq:demand_estimation} for the parameter estimation purpose, not in the simulator.\par

In the above equation, all variables excluding $\frac{\lambda_k}{\omega_k}$ are known. Thus, we use linear regression to estimate $\frac{\lambda_k}{\omega_k}$, which is subsequently used to generate dwell times associated with forward headways. $\omega_k$'s are fixed at 0.5 passengers/sec, as suggested by \citet{luo2021joint}.

\section{Determining \texorpdfstring{$\bar{H}$}{H-bar} in the FHC strategy}\label{apdx:deter_psi}

\setcounter{equation}{0}\renewcommand{\theequation}{E.\arabic{equation}}

We calculate $\bar{H}$ by solving the following equation:

\begin{equation}
	\sum_{l_k\in L}{\mathbb{E}[Y_k]} + \sum_{c_k\in C}{\mathbb{E}[\bar{H}\cdot\max\{\Lambda(\lambda_k), \Lambda(\sum\nolimits_{j}{\rho_{jk}\lambda_j}) \} ]} + \sum_{c_k\in C}{\mathbb{E}[H_k]} = M\cdot \bar{H},
\end{equation}
where the left-hand side (LHS) represents the average trip time for a bus to complete one loop of its service route. In our simulation, we disregard any layover time at the terminal stop. This implies that the bus promptly commences a new cycle of service as soon as it arrives at the terminal stop.\par

It consists of three parts. The first part is the total expected link travel time, $\sum_{l_k\in L}{\mathbb{E}[Y_k]}$, which can be directly obtained from a travel time distribution calibrated by historical data. The second part is the total expected dwell time, where $\Lambda(x)$ represents a Poisson random variable with a mean of $x$. Given $\lambda_k$'s and $\rho_{jk}$'s, the expectation of the maximum of $\Lambda(\lambda_k)$ and $\Lambda(\sum\nolimits_{j}{\rho_{jk}\lambda_j})$ is determined numerically using the Monte Carlo method. To compute the third part, the total expected holding time $\sum_{c_k\in C}{\mathbb{E}[H_k]}$, we first set $\sum_{c_k\in C}{\mathbb{E}[H_k]}$ to 0 to obtain an initial $\bar{H}$. Subsequently, we execute the FHC strategy, as defined in Eq.~\eqref{eq:Daganzo2009}, using the initial $\bar{H}$ in the simulation. The simulated holding times are then used to calculate $\sum_{c_k\in C}{\mathbb{E}[H_k]}$. The updated $\sum_{c_k\in C}{\mathbb{E}[H_k]}$ is further used to compute a new $\bar{H}$, and this iterative process continues until $\bar{H}$ converges.

\section{\texorpdfstring{Visualization of bus trajectories under the\\Offline-GMAC control}{Visualization of bus trajectories under the Offline-GMAC control}}\label{apdx:ts_visualization}

\setcounter{figure}{0}
\renewcommand{\thefigure}{F.\arabic{figure}}

Figs.~\ref{fig:visual_ts}a and b illustrate time-space trajectories of buses under two scenarios: without control and with a trained Offline-GMAC policy, for a representative episode. The color of each trajectory indicates the bus occupancy ratio (number of onboard passengers divided by bus capacity). In the uncontrolled scenario (Fig.~\ref{fig:visual_ts}a), severe bus bunching is evident, as buses struggle to maintain consistent headways throughout the day. In contrast, our trained Offline-GMAC policies, which utilize only historical uncontrolled bus trajectory data, successfully maintain uniform headways and prevent bus bunching phenomena across all time periods.

\begin{figure}[!ht]
  \centering
  \subfloat[]{%
       \includegraphics[width=0.45\textwidth]{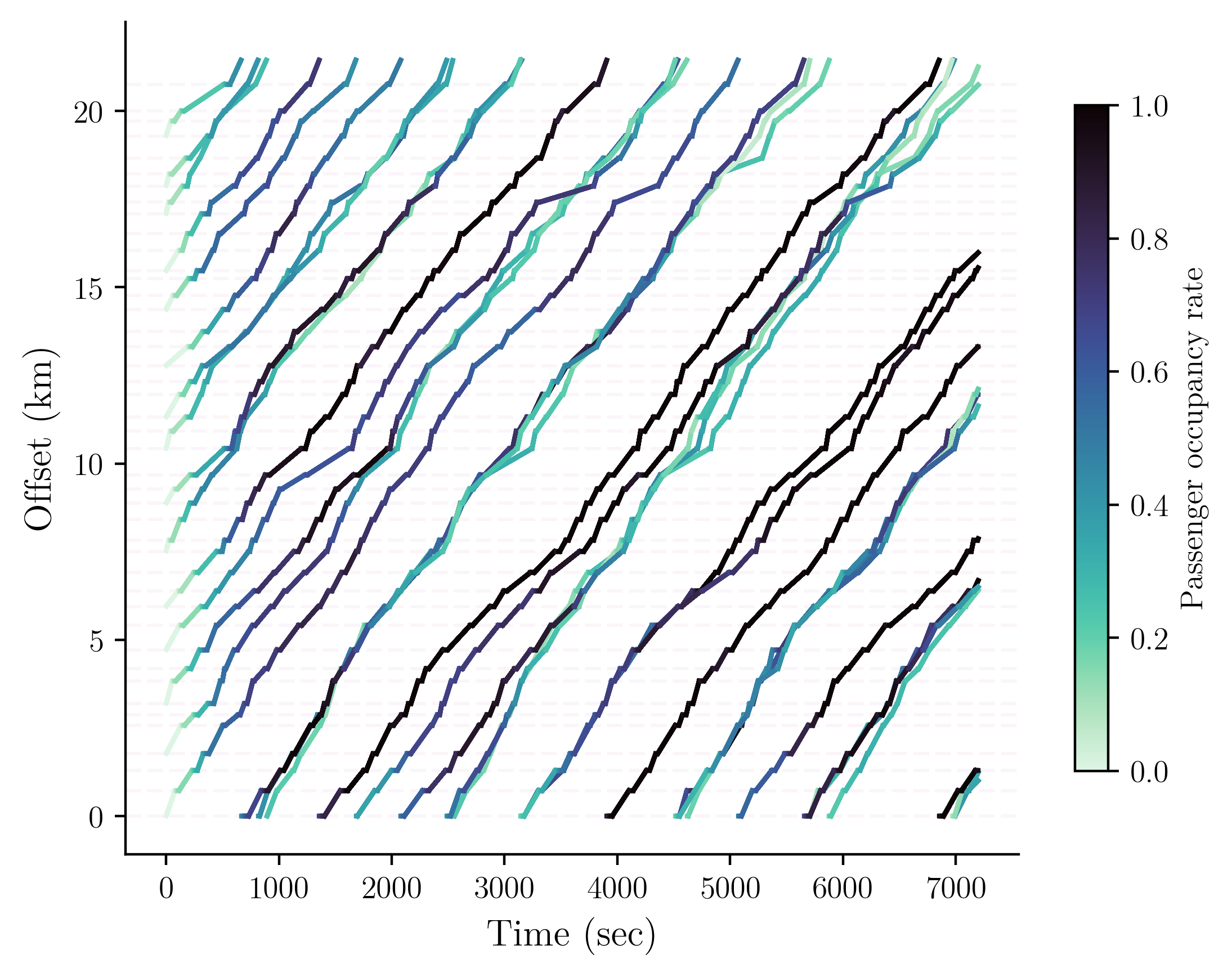}}
  \subfloat[]{%
        \includegraphics[width=0.45\textwidth]{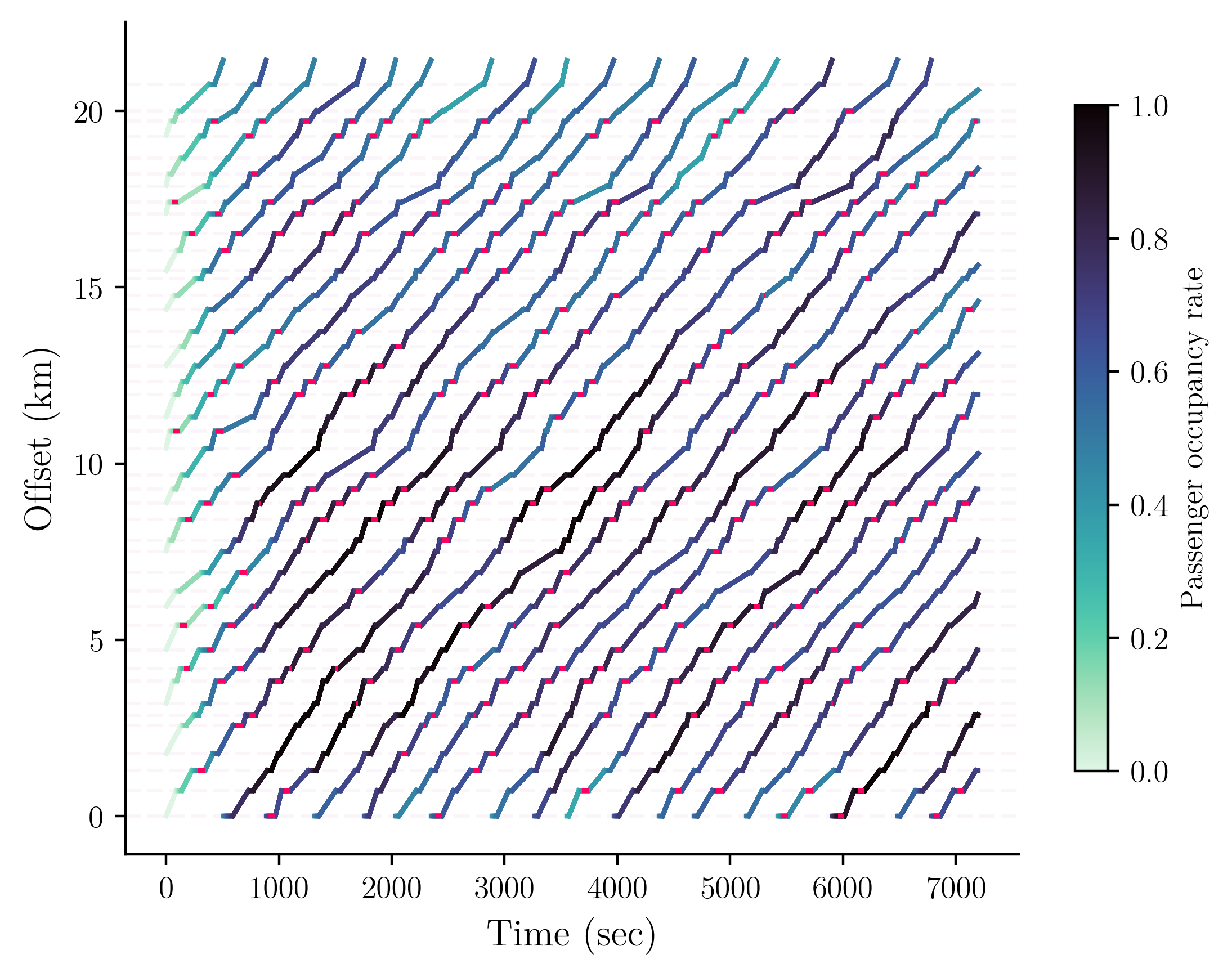}}
  \caption{Comparison of bus time-space trajectories: (a) Without control interventions. (b) Regulated by the proposed Offline-GMAC control policy.}
  \label{fig:visual_ts}
\end{figure}

\end{appendices}

\begin{footnotesize}
\bibliographystyle{informs2014trsc.bst}
\bibliography{literature.bib}
\end{footnotesize}
\end{document}